\documentclass{article}

\usepackage{fullpage}
\usepackage{mymacros}

\usepackage[numbers]{natbib}

\usepackage{authblk}
\author[1]{Sijing Tu\footnote{\url{sijing@kth.se}}}
\author[1]{Stefan Neumann\footnote{\url{neum@kth.se}}}
\author[1]{Aristides Gionis\footnote{\url{argioni@kth.se}}}
\affil[1]{KTH Royal Institute of Technology, Stockholm, Sweden}

\date{}

\begin{document}

\title{Adversaries with Limited Information in the Friedkin--Johnsen Model}

\maketitle 

\begin{abstract}
	In recent years, online social networks have been the target of adversaries
	who seek to introduce discord into societies, to undermine democracies and
	to destabilize communities. Often the goal is not to favor a certain
	side of a conflict but to increase disagreement and polarization. To get a
	mathematical understanding of such attacks, researchers use 
	opinion-formation models from sociology, such as the Friedkin--Johnsen model, 
	and formally study how much discord the adversary can produce when altering
	the opinions for only a small set of users.  In this line of work, it is
	commonly assumed that the adversary has full knowledge about the network
	topology \emph{and the opinions of all users}.  However, the latter
	assumption is often unrealistic in practice, where user opinions are not
	available or simply difficult to estimate accurately.

	To address this concern, we raise the following question: 
	\emph{Can an attacker sow discord in a social network, 
	even when only the network topology is known?} 
	We answer this question affirmatively.  We present approximation algorithms
	for detecting a small set of users who are highly influential for the
	disagreement and polarization in the network.
	We show that when the adversary radicalizes these users and if the initial
	disagreement/polarization in the network is not very high, then our method
	gives a constant-factor approximation on the setting when the user opinions
	are known.
	To find the set of influential users, we provide a novel approximation algorithm 
	for a variant of {\sc MaxCut} in graphs with positive \emph{and negative}
	edge weights. 
	We experimentally evaluate our methods, which have access \emph{only} to the
	network topology, and we find that they have similar performance as
	methods that have access to the network topology \emph{and all user
	opinions}.
	We further present an \NP-hardness proof, which was left as an open question
	by Chen and Racz~[IEEE Transactions on Network Science and Engineering, 2021].
\end{abstract}

\section{Introduction}

	Online social networks have become an integral part of modern societies and
	are used by billions of people on a daily basis.  In addition to connecting
	people with their friends and family, online social networks facilitate
	societal deliberation and play an important role in forming the political
	will in modern democracies.

	However, during recent years we have had ample evidence of
	malicious actors performing attacks on social networks so as to destabilize
	communities, sow disagreement, and increase polarization.  For instance, a
	report issued by the United States Senate finds that Russian ``trolls
	monitored societal divisions and were poised to pounce when new events
	provoked societal discord'' and that this ``campaign [was] designed to sow
	discord in American politics and society''~\cite{senatereport}. 
	Another report found that both left- and
	right-leaning audiences were targeted by these trolls~\cite{diresta2019tactics}. 
	Similarly, a recent analysis regarding the Iranian
	disinformation claimed that ``the main goal is to control public 
	opinion---pitting groups against each other and tarnishing the reputations of
	activists and protesters''~\cite{hassaniyan2022long}.
	
	The study of how such attacks influence societies can be fa\-cil\-i\-tat\-ed 
	by models of \emph{opinion dynamics}, 
	which study the mechanisms for individuals to form their opinions in social networks.
	Relevant research questions have been investigated in different disciplines, 
	e.g., psychology, social sciences, and
	economics~\cite{castellano2009statistical,jackson2008social,acemoglu2011opinion,
		xia2011opinion,lorenz2007continuous}. 
	A popular model for studying such questions in 
	computer science~\cite{gionis2013opinion,musco2018minimizing,zhu2021minimizing,bindel2015bad,abebe2021opinion,xu2021fast,tu22viral}
	is the \emph{Friedkin--Johnsen model} (FJ)~\cite{friedkin1990social}, 
	which is a generalization of the \emph{DeGroot~model}~\cite{degroot1974reaching}.

	To understand the power of an adversarial actor over the opinion-formation
	process in a social network, there are two popular measures of
	{discord}: \emph{disagreement} and \emph{polarization}; 
	for the rest of the paper, we use the word \emph{discord} 
	to refer to either disagreement or polarization; 
	see Section~\ref{sec:preliminaries} for the formal definitions.
	
	Previous works studied the increase of discord that can be
	inflicted by a malicious attacker who can change the opinions of a small
	number of users. As an example, Chen and Racz~\cite{chen2020network} showed
	that even simple heuristics, such as changing the opinions of centrists, can
	lead to a significant increase of the disagreement in the network. They also
	presented theoretical bounds, which were later extended by Gaitonde,
	Kleinberg and Tardos~\cite{gaitonde2020adversarial}.

	Crucially, the previous methods assume that the attacker has access to
	the network topology \emph{as well as} the opinions of all users. 
	However, the latter assumption is rather impractical: 
	user opinions are either not available or difficult to estimate accurately.
	On the other hand, obtaining the network
	topology is more feasible, as networks often provide access to
	the follower and interaction~graphs.

	As knowledge of all user opinions appears unrealistic, 
	we raise the following question: 
	\emph{Can attackers sow a significant amount of discord in a social network, 
	even when only the network topology is known?}
	In other words, we consider a setting \emph{with limited information} in
	which the adversary has to pick a small set of users, without knowing the
	user opinions in the network.

	\smallskip
	\noindent
	\textbf{Our Contributions.} Our main contributions are as follows.
	First, we provide a formal connection between the settings of full (all user
	opinions are known) and limited information (the user opinions are unknown).
	Informally, we show that if the variance of user opinions in the network 
	is not very high (and some other mild technical assumptions),
	then an adversary who radicalizes the users who are highly influential
	for the network obtains a $\bigO(1)$-approximation for the setting when all
	user opinions are known. Thus, we
	answer the above question affirmatively from a theoretical point of view.

	Second, we implement our algorithms and evaluate them on real-world datasets.
	Our experiments show that for maximizing disagreement, our algorithms,
	\emph{which use only topology information}, outperform simple baselines and
	have similar performance as existing algorithms that have \emph{full}
	information. Therefore, we also answer the above question affirmatively in
	practice.

	Third, we provide constant-factor approximation algorithms for identifying
	$\Omega(n)$ users who are \emph{highly influential} for the discord in the
	network, where $n$ is the number of users in the network. 
	We derive analytically the concept of highly-influential users for
	network discord and we formalize an associated computational task	
	(Section~\ref{sec:algorithms-our-problem}).
	Our formulation allows us to obtain insights into which users drive the
	disagreement and the polarization in social networks.  We also show that
	this problem is \NPhard, which solves an open problem by Chen and
	Racz~\cite{chen2020network}.

	Fourth, we show that to find the users who are influential on the discord,
	we have to solve a version of cardinality constrained {\sc Max\-Cut} in
	graphs with \emph{both positive and negative edge weights}.  For this
	problem, we present the first constant-factor approximation algorithm when
	the number of users to radicalize is $\Omega(n)$. Here the main technical challenge arises
	from the presence of negative edges, which imply that the problem is
	non-submodular and which rule out using 
	\emph{averaging arguments} that are often used to
	analyze such algorithms~\cite[A.3.2]{arora2009computational}. 
	Hence, existing algorithms do not extend to our
	more general case and we prove analogous results for graphs with positive
	and negative edge weights.  In addition, our \NP-hardness proof provides a
	further connection between maximizing the disagreement and {\sc Max\-Cut}.

	We discuss some of the ethical aspects of our findings 
	regarding the power of a malicious adversary who has access to the topology
	of a social network in the conclusion (Section~\ref{sec:conclusion}). 

\sbpara{Related work.}
A recently emerging and popular topic in the area of graph mining 
is to study optimization problems based on FJ opinion dynamics. 
Papers considered minimizing disagreement and
polarization indices~\cite{musco2018minimizing}, maximizing
opinions~\cite{gionis2013opinion}, changes of the network
topology~\cite{zhu2021minimizing,bindel2015bad} or changes of the susceptibility
to persuasion~\cite{abebe2021opinion}. Xu et al.~\cite{xu2021fast} show how to
efficiently estimate quantities such as the polarization and disagreement indices.
Our paper is also conceptually related to the topic of maximizing influence in
social networks, pioneered by Kempe, Kleinberg and
Tardos~\cite{kempe2015maximizing}; 
the influence-maximization model has recently been
combined with opinion-formation processes~\cite{tu22viral}.
Furthermore, many extensions of the classic FJ model have been
proposed~\cite{amelkin17polar,parsegov2016novel}.

Most related to our work are the papers by Chen and Racz~\cite{chen2020network}
and by Gaitonde, Kleinberg and Tardos~\cite{gaitonde2020adversarial}, who
consider adversaries who plan network attacks.  They provide upper bounds when
an adversary can take over $k$~nodes in the network, and they present heuristics
for maximizing disagreement in the setting with full information.  A practical
consideration of this model has motivated us to study settings with limited
information.  While their adversary can change the opinions of $k$~nodes to
either $0$ or $1$, in this paper we are mainly concerned with adversaries which
can change the opinions of $k$~nodes to $1$; we consider the adversary's
actions as ``radicalizing $k$~nodes.''  Our setting is applicable in scenarios
when opinions near opinion value $0$ ($1$) correspond to non-radicalized 
(radicalized) views.

Our algorithm for {\sc MaxCut} in graphs with positive and negative edge weights 
is based on the SDP-rounding techniques by Goemans and Williamson~\cite{goemans1995improved} for {\sc MaxCut}, 
and by Frieze and Jerrum~\cite{frieze1997improved} for {\sc Max\-Bi\-section}. 
While their results assume
that the matrix~$\m+A$ in Problem~\eqref{eq:our-problem} is the Laplacian of a
graph with positive edge weights, our result in Theorem~\ref{thm:unbalanced}
applies to more general matrices, %
albeit with worse approximation ratios.
Currently, the best approximation algorithm for {\sc Max\-Bi\-section} is by Austrin et al.~\cite{austrin2016better}.  
Ageev and Sviridenko~\cite{ageev1999approximation} 
gives LP-based algorithms for
versions of {\sc Max\-Cut} with given sizes of parts,
Feige and Langberg~\cite{feige2001approximation} extended this work to an SDP-based algorithm;
but their techniques appear to be
inherently limited to positive edge-weight graphs and cannot be extended to our more general
setting of Problem~\eqref{eq:our-problem}.

\section{Preliminaries}
\label{sec:preliminaries}

Let $\graph = (V, E, w)$ be an undirected weighted graph 
representing a social network.
The edge-weight function 
$w \colon E \rightarrow \Real_{>0}$
models the strengths of user interactions. 
We write $\abs{V} = n$ for the number of nodes, %
and use $N(u)$ to denote the set of neighbors of node $u\in V$, 
i.e., $N(u) = \{v : (u, v) \in E\}$. 
We let $\m+D$ be the $n\times n$ diagonal matrix with
$\m+D_{v,v} = \sum_{u\in N(v)}w(u,v)$ and 
define the weighted adjacency matrix $\m+W$
by $\m+W_{u,v} = w(u,v)$. %
The Laplacian of the graph $G$ is given by $\laplacian = \m+D - \m+W$.

In the \emph{Friedkin--Johnsen opinion-dynamics model} (FJ)~\cite{friedkin1990social}, 
each node $u\in V$ corresponds to a person who has
an \emph{innate opinion} and an \emph{expressed opinion}.
For each node~$u$, the innate opinion $s_u\in[0,1]$ is fixed over time and kept
private; the expressed opinion $z_u^{(t)}\in[0,1]$ is publicly known 
and it changes over time $t\in\mathbb{N}$ due to peer pressure. 
Initially, $z_u^{(0)}=s_u$ for all users $u\in V$. 
At each time
$t>0$, all users $u\in V$ update their expressed opinion $z_u^{(t)}$ as the
weighted average of their innate opinion and the expressed opinions of their
neighbors, as 
follows:
\begin{align}
\label{eq:update-opinions}
	z_u^{(t)}
	= \frac{s_u + \sum_{v\in N(u)} w_{uv} z_v^{(t-1)}}{1 + \sum_{v\in N(u)} w_{uv}}.
\end{align}

We write $\v+z^{(t)} = (z_1^{(t)},\dots,z_n^{(t)})$
to denote the vector of expressed opinions at time $t$.
Similarly, we set $\v+s = (s_1,\dots,s_n)$ for the innate opinions.
In the limit $t\to\infty$, the expressed opinions
reach the equilibrium $\finop = \lim_{t\to\infty} \v+z^{(t)} = (\ID + \laplacian)^{-1} \begop$. 

We study the behavior of the following two \emph{discord measures} in the FJ
opinion-dynamics model:
\begin{description}
\item \textbf{Disagreement index}~(\cite{musco2018minimizing})
	$\DisIdx{G, \begop} = \sum_{(u,v)\in E} w_{u,v} (z_u-z_v)^2 
	= \begop^{\intercal} \MasIdx{\DisIdx{}} \begop$, where
	$\MasIdx{\DisIdx{}} = (\laplacian + \ID)^{-1} \laplacian (\laplacian +
			\ID)^{-1}$, and
\item \textbf{Polarization index}~(\cite{matakos2017measuring, musco2018minimizing})
	$\PolIdx{G, \begop} = \sum_{u\in V} (z_u - \bar{\finop})^2 =
	\begop^{\intercal} \MasIdx{\PolIdx{}} \begop$, where
	$\bar{\finop} = \frac{1}{n} \sum_{u\in V} z_u$ is the average user
	opinion and 
	$\MasIdx{\PolIdx{}} = (\ID + \laplacian)^{-1} (\ID - \frac{\ind \ind^\intercal}{n}) (\ID + \laplacian)^{-1}$.
\end{description}
Note that the disagreement index measures the discord along the edges of the network,
i.e., it measures how much interacting nodes disagree. The polarization index measures
the overall discord in the network by considering the variance of the opinions.

We note that the matrices $\MasIdx{\DisIdx{}}$ and $\MasIdx{\PolIdx{}}$ may have
positive and negative off-diagonal entries and it is not clear whether they are
diagonally dominant; this is in contrast to graph Laplacians that have
exclusively non-positive off-diagonal entries and are diagonally dominant.
Having positive and negative entries will be one of the challenges we need to
overcome later.  The following lemma presents some additional properties.
\begin{lemma}
\label{lem:matrices}
	Let $\m+A\in\{\MasIdx{\DisIdx{}}, \MasIdx{\PolIdx{}}\}$.
	Then $\m+A$ is positive semi\-definite and satisfies $\m+A \v+1 = \v+0$, where $\v+1$ is the all-ones
	vector and $\v+0$ is the all-zeros vector.
\end{lemma}

While we consider opinions in the interval $[0,1]$, 
for the \NPhardness results presented later, 
for technical reasons,
it will be useful to consider opinions in the interval~$[-1,1]$.
In Appendix~\ref{sec:scaling}, we show that the solutions 
of optimization problems are maintained 
under scaling, which implies that our \NPhardness results and 
our $\bigO(1)$-approximation algorithm also apply for opinions in the $[0, 1]$~interval.

We present all omitted proofs in Appendix~\ref{sec:omitted-proofs}. 

\section{Problem definition and algorithms}
\label{sec:algorithms}

We start by defining the problem of maximizing the discord when~$k$ user
opinions can be \emph{radicalized}, i.e., when for $k$~users the innate opinions
can be changed from their current value $\begop_0(u)$ to the extreme value $1$.
This problem is of practical relevance when opinions close to~$0$ correspond to
non-radicalized opinions (``{\sc covid-19} vaccines are generally safe'') and
opinions close to $1$ correspond to radicalized opinions (``{\sc covid-19}
vaccines are harmful'').  Then an adversary can radicalize $k$~people by setting
their opinion to $1$, for instance, by supplying them with fake news or by
hacking their social network accounts.  Formally, our problem is stated as
follows.

\begin{problem}
	Let $\m+A\in\{\MasIdx{\DisIdx{}}, \MasIdx{\PolIdx{}}\}$.
	Consider an undirected weighted graph $G = (V, E, w)$, and innate opinions
	$\begop_0\in[0, 1]^n$.
	We want to maximize the discord
    where we can radicalize the innate opinions of $k$~users.
	In matrix notation, the problem is as follows:
	\begin{equation}
    \label{problem:max-disagreement}
		\begin{aligned}
			\max_{\begop} \quad &  \begop^{\intercal} \m+A \begop,\\
			\st \quad &\lVert \begop - \begop_0 \rVert_0 = k, \text{ and}\\
			& \begop(u) \in \{\begop_0(u), 1\} \text{ for all } u\in V.
		\end{aligned}
	\end{equation}   
\end{problem}

Note that if we set $\m+A = \MasIdx{\DisIdx{}}$ the problem is to
maximize the disagreement in the network. If we set $\m+A = \MasIdx{\PolIdx{}}$
we seek to max\-i\-mize the polarization.

Further observe that for Problem~\eqref{problem:max-disagreement}, the algorithm
obtains as input the graph~$G$ \emph{and the vector of innate opinions
$\begop_0$}.  Therefore, we view this formulation as the setting \emph{with full
information}.

Central to our paper is the idea that the algorithm has access to the topology of the graph~$G$,
\emph{but it does not have access to the initial innate opinions~$\begop_0$}.
As discussed in the introduction, we believe that this scenario is of higher
practical relevance, as it seems infeasible for an attacker to gather the
opinions of millions of users in online social networks.  On the other hand,
assuming access to the network topology, i.e., the graph~$G$, appears more
feasible because networks, such as Twitter, make this information publicly
available. 

Our approach for maximizing the discord, even when we have limited information,
i.e., we only have access to the graph topology, has two steps:
\begin{enumerate}
	\item[{\bf 1.}] Detect a small set~$S$ of $k$~users who are highly influential
		for the discord in the network.
	\item[{\bf 2.}] Change the innate opinions for the users in the set~$S$ to $1$ and
		leave all other opinions unchanged.
\end{enumerate}

\smallskip
In the rest of this section, we will describe our overall approach for finding a
set of $k$~influential users on the discord
(Section~\ref{sec:algorithms-our-problem}) and then we will discuss
approximation algorithms (Section~\ref{sec:balanced-maxcut}) and heuristics
(Section~\ref{sec:greedy}) for this task. 
Then, we prove computational hardness (Section~\ref{sec:hardness}).

\subsection{Finding influential users on the discord}
\label{sec:algorithms-our-problem}

Next, we describe the implementation of Step~(1) discussed above. 
In other words, we wish
to find a set~$S$ of $k$~users who are highly influential for the discord in the
network.

To form an intuition about highly-influential users for the network discord in
the absence of information about user innate opinions, we consider scenarios of
non-controversial topics. Since the topics are non-controversial, we expect most
users to have opinions near a consensus opinion~$c$. %
In such scenarios, an adversary who aims to radicalize $k$~users so as to maximize
the network discord, will seek to find a set~$S$ of $k$~users and set
$\begop_0(u)=1$, for $u\in S$, so as to maximize the discord
$\begop^{\intercal} \m+A \begop$, where
$\m+A\in\{\MasIdx{\DisIdx{}}, \MasIdx{\PolIdx{}}\}$.

Since we assume most opinions to be near consensus~$c$, it seems natural that
the concrete value of~$c\in[0,1]$ has no big effect on the choice of the users
picked by the adversary (see also Theorem~\ref{thm:relationship} which
formalizes that this intuition is correct).  Hence, we consider $c=0$ and study
the idealized version of the problem, where $\begop_0(u)=0$ for all users $u$ in
the network.  In this case, the adversary will need to solve the following
optimization problem:
\begin{align}
\begin{aligned}
\label{eq:our-problem}
	\max_{\begop} \quad & \begop^{\intercal} \m+A \begop,\\
	\st \quad &\lVert \begop \rVert_0 = k, \text{ and} \\
		& \begop \in [0,1]^n.
\end{aligned}
\end{align}
The result of the above optimization problem is a vector~$\v+s$ that
has $k$ non-zero entries, all of which are equal to $1$. %
Thus, we can view the set $S=\{u \mid \begop_0(u)=1\}$ as a set of users who are
highly influential for the discord in the network.  

We provide a constant-factor approximation algorithm for this problem in
Theorem~\ref{thm:unbalanced}.  We also show that the problem is \NPhard in
Theorem~\ref{thm:disagreement-np-hard} when $\m+A=\MasIdx{\DisIdx{}}$, which
answers an open question by Chen and Racz~\cite{chen2020network}.

\sbpara{Relationship between the limited and full information settings.}
At first glance, it may not be obvious why a solution for
Problem~\eqref{eq:our-problem} with limited information implies a good solution
for Problem~\eqref{problem:max-disagreement} with full information. However, we
will show that this is indeed the case when there is little initial discord in
the network; we believe this is the most interesting setting for attackers
who wish to increase the discord.

Slightly more formally (see Theorem~\ref{thm:relationship} for
details), we show the following. If initially all innate opinions are close to
the average opinion $c$ and some mild assumptions hold, then an
$\bigO(1)$-approximate solution for Problem~\ref{eq:our-problem} (when only the
network topology is known) implies an $\bigO(1)$-approximate solution for
Problem~\ref{problem:max-disagreement} (when full information including user
opinions are~known).

Before stating the theorem, we define some additional notation.  For a set of
users $X\subseteq V$, we write $\begop_X$ to denote the vector of innate
opinions when we radicalize the users in $X$, i.e., $\begop_X\in[0,1]^n$
satisfies $\begop_X(u) = 1$ if $u\in X$ and $\begop_X(u) = \begop_0(u)$  if
$u\not\in X$.  Furthermore, given $X$ and a vector $\v+v$, we write
$\v+v_{\vert X}$ to denote the restriction of $v$ to the entries in $X$, i.e.,
$\v+v_{\vert X}(u) = \v+v(u)$ if $u\in X$ and 
$\v+v_{\vert X}(u) = 0$ if $u\not\in X$.  We discuss our technical assumptions
after the theorem.

\begin{theorem}
\label{thm:relationship}
	Let $\m+A\in\{\MasIdx{\DisIdx{}}, \MasIdx{\PolIdx{}}\}$.
	Let $c \in [0,1)$ and $\epsvec \in [-c,1-c]^n$ be such that
	$\begop_0 = c\v+1 + \epsvec$.
	Let $\gamma_1,\gamma_2, \gamma_3 \in (0,1)$ be parameters.  Furthermore, assume
	that for all sets $X\subseteq V$ with $\abs{X}=k$ it holds that:
	\begin{enumerate}
		\item[{\bf 1.}] $(\begop_X - \begop_0)^\intercal \m+A \begop_0 
				\geq - \gamma_1 \begop_0^\intercal \m+A \begop_0$,
		\item[{\bf 2.}] $\epsvec_{\vert X}^\intercal \m+A \epsvec_{\vert X} 
				\leq \gamma_2 \begop_0^\intercal \m+A \begop_0$, and 
		\item[{\bf 3.}] $\abs{\epsvec_{\vert X}^\intercal \m+A \v+1_{\vert X}}
				\leq \gamma_3 \begop_0^\intercal \m+A \begop_0$.
	\end{enumerate}
	Suppose we have access to a $\beta$-approximation algorithm for
	Problem~\eqref{eq:our-problem} with limited information.
	Then we can compute a solution for Problem~\eqref{problem:max-disagreement}
	with full information with approximation ratio
	$\frac{1}{4} \min\{\beta,
		   \frac{ 1-2\gamma_1-2(1-c)\gamma_3 }{ 1+2(1-c)\gamma_3+\gamma_2 } \}$,
	even if we only have access to the graph topology (but \emph{not} the user
	opinions).
\end{theorem}

One may think of $c$ as the average
innate opinion and $\epsvec$ as the vector that indicates how much each innate
opinion deviates from $c$.  Indeed, for topics that initially have little
discord, one may assume that most entries in $\epsvec$ are small.

The intuitive interpretation of the technical conditions from the theorem is as
follows.
Condition~(1) corresponds to the assumption that no matter which
$k$~users the adversary radicalizes, the discord will not \emph{drop} by
more than a $\gamma_1$-fraction. This rules out some unrealistic scenarios in
which, for example, all but $k$~users have initial innate opinion~$1$ and one
could subsequently remove the entire discord by radicalizing the remaining 
$k$ users. 
Conditions~(2) and~(3) are
of similar nature and essentially state that if only $k$~users have the opinions
given by~$\epsvec$ and all other users have opinion~$0$, then the discord
in the network is significantly smaller than the initial discord when all
$n$~users have the opinions in $\begop_0$.

Note that when $k\ll n$, it is reasonable to assume that
$\gamma_1$, $\gamma_2$, $\gamma_3$ are upper-bounded by a small constant, 
say $\frac{1}{5}$. In this case the theorem states
that if we have a $\beta$-approximation algorithm for the setting with
limited information then we obtain an $\bigO(\beta)$-approx\-i\-ma\-tion algorithm
for the setting with full information, even though we only use the network
topology but not the innate opinions.

\subsection{$\alpha$-Balanced MaxCut}
\label{sec:balanced-maxcut}

In this section, we study the {$\alpha$-{\sc Balanced-Max\-Cut}} problem for
which we present a constant-factor approximation algorithm. 
This algorithm allows us to solve Problem~\eqref{eq:our-problem} (maximizing
		discord with limited information) approximately. Combined with
Theorem~\ref{thm:relationship} above, this implies that (under some assumptions)
adversaries with limited information only perform a constant factor worse than
those with full information (see Corollary~\ref{cor:relationship} below).

In the $\alpha$-{\sc Balanced-Max\-Cut} problem, the goal is to
partition a set of nodes into two sides such that one side contains an
$\alpha$-fraction of the nodes and the cut is maximized. Formally, we are given
a positive semi\-definite matrix $\m+A\in\mathbb{R}^{n\times n}$ and a parameter
$\alpha\in[0,1]$.
The goal is to solve the following problem:
\begin{equation}
\label{problem:maxcut-unbalanced}
    \begin{aligned}
        \max_{\v+x} \quad \frac{1}{4} & \v+x^{\intercal} \m+A \v+x ,\\
        \st \quad & \lVert \v+x + \v+1 \rVert_0 = \alpha n, \text{ and} \\
				  & \v+x \in [-1,1]^n.
    \end{aligned}
\end{equation}

Note that the optimal solution vector $\v+x$ takes values in $\{-1,1\}$ (since
the objective function is convex, as $\m+A$ is positive semi\-definite) and thus
it partitions the set~$V$ into two sets $S = \{ u \colon \v+x_u = 1 \}$ and
$\bar{S} = \{ u \colon \v+x_u = -1 \}$.  The first constraint ensures $\abs{S}=
\alpha n$ and $\abs{\bar{S}}= (1-\alpha)n$, i.e., one side contains an
$\alpha$-fraction of the nodes and the other side contains a
$(1-\alpha)$-fraction.  If $\m+A$ is the Laplacian of a graph and
$\alpha=\frac{1}{2}$, this is the classic {\sc Max\-Bi\-section}
problem~\cite{frieze1997improved}.  Hence, we will sometimes refer to $\v+x$ and
the corresponding partition $(S,\bar{S})$ as a \emph{cut} and to
$\v+x^{\intercal} \m+A \v+x$ as the \emph{cut value}.  

Our main result for Problem~\eqref{problem:maxcut-unbalanced} is as follows.
\begin{theorem}
\label{thm:unbalanced}
	Suppose $\alpha\in[0,1]$ is a constant and $\m+A$ is a symmetric,
	positive semi\-definite matrix with $\ind^{\intercal} \m+A = \m+0^{\intercal}$.
	Then for any $0< \epsilon <1$, there exists a randomized polynomial time algorithm 
	that with probability at least $1-\epsilon$, outputs a solution for the
	$\alpha$-{\sc Balanced-Max\-Cut} (Problem~\eqref{problem:maxcut-unbalanced}) with 
	approximation ratio presented in Figure~\ref{fig:approx-all}.
	The algorithm runs in time $O(1/\epsilon \log(1/\epsilon)poly(n))$.
\end{theorem}

\begin{figure}[t]
	\centering 
        \resizebox{0.40\textwidth}{!}{%
  \tikzsetnextfilename{tikz/approx_ratio}%
  \begin{tikzpicture}
\begin{axis}[
    legend cell align={left},
    legend columns=1,
    legend style={
        fill opacity=0,
        draw opacity=1,
        text opacity=1,
        anchor=north east,
        at={(1.5,0.85)},
        draw=none,
    },
    axis lines = left, 
    label style={font=\fontsize{14}{6}\selectfont},
    tick label style={font=\fontsize{14}{6}\selectfont}, 
    legend style={font=\fontsize{14}{6}\selectfont},
    scaled ticks=true, 
    xlabel style={at={(1,0.1)}, anchor=south west},
    xlabel={$\alpha$}, 
    ylabel={Approximation Ratio}, 
    every axis plot post/.append style={mark=none}, 
]
\addplot table [x index=0, y index=1, col sep=space] {tikz/real_ratio.txt};
\addlegendentry{Numerical Results} 

\addplot[dashed, domain=0:0.448]{2.059*x^2}; 
\addlegendentry{$2.059\cdot\alpha^2$} 

\addplot[dashed, color=red, domain=0.448:0.5]{1.36*(1-x)^2}; 
\addlegendentry{$1.36\cdot(1-\alpha)^2$} 

\addplot[dashdotted, domain=0.5:0.552]{1.36*x^2}; 
\addlegendentry{$1.36\cdot\alpha^2$} 

\addplot[dashdotted, color=red, domain=0.552:1]{2.059*(1-x)^2}; 
\addlegendentry{$2.059\cdot(1-\alpha)^2$} 

\end{axis}
\end{tikzpicture}%

		}\\
	\caption{
	The approximation ratio from Theorem~\ref{thm:unbalanced} as function
	of~$\alpha$. We present numerical approximation ratio results compared
	with a piece-wise quadratic function defined by the formulas: 
	$2.059\alpha^2$ for $0 < \alpha < 0.448$, 
	$1.36(1-\alpha)^2$ for $0.448\leq \alpha < 0.5$, 
	$1.36\alpha^2$ for $0.5 \leq \alpha < 0.552$, 
	and $2.059(1-\alpha)^2$ for $0.552 \leq \alpha < 1$. }
	\label{fig:approx-all}
\end{figure}
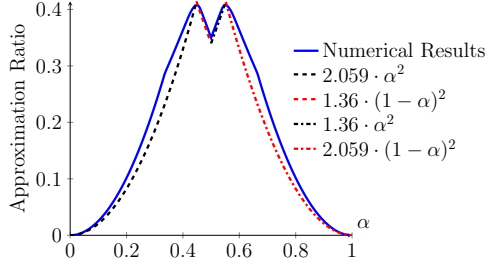

In Figure~\ref{fig:approx-all} we visualize the approximation ratios for
different values of~$\alpha$. In particular, observe that for any constant
$\alpha\in(0,1]$, our approximation ratio is $\Omega(1)$ and that for most values
of $\alpha$ it performs within at a least a factor of~$10$ of the optimal
solution. Furthermore, for $\alpha$ close to $0.5$, our approximation is better
than $\frac{1}{3}$.

Before we discuss Theorem~\ref{thm:unbalanced} in more detail, we first present
two corollaries.  First, we observe that the theorem implies that we obtain a
constant-factor approximation algorithm for maximizing the discord with limited
information (Problem~\eqref{eq:our-problem}).
\begin{corollary}
\label{cor:limited-information}
	Let $k = \alpha n$. If $\alpha\in[0,1]$ is a constant, there exists an $\bigO(1)$-approximation
	algorithm for Problem~\eqref{eq:our-problem} with limited information that
	runs in polynomial time.
\end{corollary}
\begin{proof}
	Observe that a solution for Problem~\eqref{problem:maxcut-unbalanced}
	implies a solution for Problem~\eqref{eq:our-problem} as follows:
	Suppose in Problem~\eqref{problem:maxcut-unbalanced} we set
	$\m+A=4\MasIdx{\DisIdx{}}$ or $\m+A=4\MasIdx{\PolIdx{}}$ and $\alpha =
	\frac{k}{n}$ to obtain a solution $\v+x$. Now we define a solution $\begop$
	for Problem~\eqref{problem:maxcut-unbalanced} by setting $\begop(u)=1$ if
	$\v+x(u)=1$ and $\begop(u)=0$ if $\v+x(u)=-1$.  Observe that
	$\lVert\begop\rVert_0=k$ as desired. Since the last step can be viewed as
	rescaling opinions from $[-1,1]$ to $[0,1]$, the objective function values
	of $\v+x^\intercal \m+A \v+x$ and $\begop^\intercal\m+A\begop$ only
	differ by a factor of~$4$ (see Appendix~\ref{sec:scaling}).
\end{proof}

Second, observe that by combining Theorems~\ref{thm:relationship} and
Corollary~\ref{cor:limited-information}, we immediately obtain the following
result for solving the setting with full information, even when we only have
access to the network topology.
\begin{corollary}
\label{cor:relationship}
	Suppose $\alpha\in[0,1]$ is a constant and the conditions of
	Theorem~\ref{thm:relationship} hold with $\gamma_1,\gamma_2,\gamma_3\leq
	\frac{1}{5}$.  Then there exists a polynomial time algorithm for
	Problem~\eqref{problem:max-disagreement} with full information that has
	approximation ratio~$\bigO(1)$ and \emph{only uses the graph topology}
	(but \emph{not} the user opinions).
\end{corollary}

Next, let us discuss Theorem~\ref{thm:unbalanced} in more detail.

\smallskip
The theorem \emph{generalizes} previous results and its
approximation ratios are only a small constant factor worse than classic
results~\cite{goemans1995improved,frieze1997improved,han02improved}. In
particular, the previous results assumed that $\m+A$ is the Laplacian of a graph
with positive edge weights, and thus $\m+A$ has the structure that all
off-diagonal entries are non\-positive. In contrast, in our result we do require
the latter assumption and allow for positive off-diagonal
entries, which appear, for instance, in graphs with negative edge weights.
Indeed, this is the case for the matrices~$\MasIdx{\DisIdx{}}$ and
$\MasIdx{\PolIdx{}}$ from Section~\ref{sec:preliminaries}, which may have
positive off-diagonal entries.  Therefore, our generalized theorem is necessary
to maximize the discord in Problem~\eqref{eq:our-problem}.

Furthermore, we note that to apply Theorem~\ref{thm:unbalanced} on graphs with
both positive \emph{and negative} edge weights, we have to assume that their
Laplacian is positive semi\-definite. This assumption cannot be dropped, as
pointed out by Williamson and Shmoys~\cite[Section~6.3]{williamson2011design}.
This is crucial since, while for graphs with positive edge weights the Laplacian
is always positive semi\-definite, this is not generally true for graphs with
negative edge weights.\footnote{
	We note that this property of graphs with positive and negative edges also
	rules out simple algorithms of the type: ``Randomly color the graph, pick
	the color for which the induced subgraph has the highest edge weights and
	then solve unconstrained MaxCut in this subgraph.'' The issue here is that
	the Laplacian of such a randomly picked subgraph is not necessarily positive
	semidefinite (even if the Laplacian of the original graph is positive
	semidefinite). Thus such simple tricks cannot be applied here and we need
	other solutions.
} However, this assumption holds in our use cases due to
Lemma~\ref{lem:matrices}.

In the theorem we require $\alpha\in[0,1]$ to be a constant and thus
$k=\Omega(n)$. While this is somewhat undesirable, there are underlying
technical reasons for it: the SDP-based approach by Frieze and
Jerrum~\cite{frieze1997improved} also has this requirement; LP-based algorithms
which work for $k=o(n)$ (as shown, for instance, by
\citet{ageev1999approximation}) do not generalize to the setting in which the
matrices are not graph Laplacians; the same is the case for the SDP-based
approach by \citet{feige2001approximation}.

\sbpara{Algorithm.} Our algorithm is based on solving the SDP relaxation of
Problem~\eqref{problem:maxcut-unbalanced} and applying random hyperplane
rounding~\cite{goemans1995improved}, followed by a greedy step in which we adjust
the sizes of the sets $S$ and $\bar{S}$. Later, we will see that our main
technical challenge will be to prove that the greedy adjustment step still works
in our more general setting.

To obtain our SDP relaxation of Problem~\eqref{problem:maxcut-unbalanced},
we observe that by the convexity of the objective function we can assume that
$\begop\in\{-1,1\}^n$ (see Appendix~\ref{sec:convexity}) 
and thus, we can rewrite the constraint
$\lVert \v+x + \v+1 \rVert_0 = \alpha n$ 
~as~ $2\sum_{i<j}^n \v+x_i \v+x_j = n^2(1 - 2\alpha)^2 - n$.

Now the semi\-definite relaxation of Problem~\eqref{problem:maxcut-unbalanced}
becomes:
\begin{align}
\begin{aligned}
\label{problem:maxcut-unbalaced-relax}
	\max_{\v+v_1,\dots,\v+v_n} \quad & \frac{1}{4} \sum_{ij} \m+A_{ij}\v+v_i^{\intercal} \v+v_j,\\
	\st \quad & \sum_{i<j} \v+v_i^{\intercal} \v+v_j = \frac{1}{2}n^2(1 - 2 \alpha)^2-\frac{n}{2}, \text{ and} \\
	& \v+v_i \in \Real^n, \quad\quad \|\v+v_i\|_2 = 1.
\end{aligned}
\end{align}

{\SetAlgoNoLine
 \LinesNumbered
 \DontPrintSemicolon
 \SetAlgoNoEnd
\begin{algorithm2e}[t]
Solve the SDP in Equation~\eqref{problem:maxcut-unbalaced-relax}
		with solution $\v+v_1,\dots,\v+v_n$\;
\For{$\kappa = \bigO(1/\varepsilon \lg(1/\varepsilon))$ times}{
	 Sample vector $\v+r$ with each entry $\sim {\mathcal N}(0,1)$\;
	 Set $S = \{ i : \langle \v+v_i, \v+r \rangle \geq 0 \}$ and $\bar{S}=V\setminus S$\;
	Set %
			$\bar{\v+x}_i=1$ if $i\in S$ and 
			$\bar{\v+x}_i=-1$ if $i\in \bar{S}$\;
	\If{$\abs{S}>\alpha n$}
		{greedily move elements from $S$ to $\bar{S}$ until $\abs{S} = \alpha n$}
	\If{$\abs{S}<\alpha n$}
		{greedily move elements from $\bar{S}$ to $S$ until $\abs{S} = \alpha n$}
	\tcp*{\small In each step, move the element that
			decreases the value of the objective function the least}
}
\Return{best solution over all trials}\;
\caption{SDP-based relaxation followed by iterative local improvement}
\label{alg:SDP-based}
\end{algorithm2e}
}

Our approach for solving Problem~\eqref{problem:maxcut-unbalanced}
is shown as Algorithm~\ref{alg:SDP-based}.
For simplicity, we assume that $\alpha\leq \frac{1}{2}$; 
for $\alpha>\frac{1}{2}$ we can run the algorithm with $\alpha'=1-\alpha\leq \frac{1}{2}$
and obtain the desirable result.

\sbpara{Analysis.}
Our analysis has two parts. The first part is the
hyperplane rounding of the SDP solution; it follows the techniques of
Goemans and Williamson~\cite{goemans1995improved} and 
Frieze and Jerrum~\cite{frieze1997improved}. 
The next lemma summarizes the first part of the analysis.
\begin{lemma}[\cite{goemans1995improved,williamson2011design,frieze1997improved}]%
\label{lem:standard}
	The expected cut of $(S, \bar{S})$ is at least
	$\frac{2}{\pi}\, \OPT$ and 
	$\Exp[\abs{S}\abs{\bar{S}}] \geq 0.878 \cdot \alpha(1-\alpha)n^2$, 
	where $\OPT$ is the optimal solution for 
	$\alpha$-{\sc Balanced-Max\-Cut}.
\end{lemma}

The second part of the analysis is novel and considers the greedy procedure that
ensures that $S$ contains~$k$ elements. 
When $\m+A$ is the Laplacian of a graph with \emph{non\-negative} edge weights, 
an averaging argument (see, e.g., ~\cite[A.3.2]{arora2009computational}) implies
that there exists $u\in S$ such that we can move $u$ from $S$ to $\bar{S}$ and
the cut value drops by a factor of at most $1/\abs{S}$.  
However, for more general matrices $\m+A$ this may not hold, 
e.g., when $\m+A$ is the Laplacian of a \emph{signed graph} with negative edge
weights or when $\m+A$ is the matrix that corresponds to the disagreement index,
		as in Problem~\eqref{eq:our-problem}.
We also illustrate this in Appendix~\ref{sec:illustration-worst-case}. 
However, we show that in our setting there always
exists a node in~$S$ such that if we move $u$ from $S$ to $\bar{S}$ 
then the cut value drops by a factor of at most $2/\abs{S}$.

\begin{lemma}
    \label{lemma:boundSUmodify}
	Suppose that $\m+A$ is a symmetric, positive semi\-definite matrix with
	$\v+1^{\intercal} \m+A = \v+0^{\intercal}$ and let $\vx\in\{-1,1\}^n$.
    Set $M = \frac{1}{4}\vx^{\intercal} \m+A \vx$ and
	$S = \{i \in \{1, \ldots, n\} \mid \vx_i = 1\}$.
    Then there exists $i \in S$ such that, by modifying $\x_i$ to be $-1$, 
    $M$ decreases at most $\frac{2M}{\abs{S}}$.
\end{lemma}
\begin{proof}
    We prove the lemma by contradiction. 
    Suppose there does not exist such $i \in S$ and
    let $\v+e_i \in \Real^n$ denote the vector whose $i$-th entry is $1$ and all other entries are $0$s.   
    Then for any $i$, it holds 
    \begin{align*}
		M - \frac{1}{4}(\vx - 2\v+e_i)^{\intercal} \m+A (\vx - 2\v+e_i) > \frac{2M}{\abs{S}}.
	\end{align*}
    Expanding and simplifying the formula, we get
	$$\v+e_i^{\intercal} \m+A \vx > \frac{2M}{\abs{S}} + \v+e_i^\intercal \m+A \v+e_i,$$ 
	for all $i\in S$.
	Summing this inequality over all $i\in S$, we obtain
    \begin{align}
	\label{eq:boundSUmodify-1}
		 \sum_{i\in S} \v+e_i^{\intercal} \m+A \vx > 2M + \sum_{i\in S} \v+e_i^\intercal \m+A \v+e_i.
	\end{align}

    By $\ind^{\intercal} \m+A = \m+0^{\intercal}$, and $\sum_{i\in S} \v+e_i + \sum_{i \in U} \v+e_i = \ind$, 
    we get
    \begin{equation*}
        \sum_{i\in S} \v+e_i^{\intercal} \m+A \vx + \sum_{i\in \bar{S}} \v+e_i^{\intercal} \m+A \vx
		= \ind^{\intercal} \m+A \vx 
		= 0. 
    \end{equation*}

	Thus,
	\begin{align}
        \label{eq:boundSUmodify-2}
		\sum_{i\in S} \v+e_i^{\intercal} \m+A \vx = - \sum_{i\in \bar{S}} \v+e_i^{\intercal} \m+A \vx.
	\end{align}

    Using $\sum_{i \in S} \v+e_i - \sum_{i \in \bar{S}} \v+e_i = \vx$ and Equation~\eqref{eq:boundSUmodify-2},
    we get 
    \begin{align}
	\label{eq:boundSUmodify-3}
        \vx^{\intercal} \m+A \vx 
		= \left(\sum_{i \in S} \v+e_i - \sum_{i \in \bar{S}} \v+e_i\right)^{\intercal} \m+A \vx
		= 2\sum_{i\in S} \v+e_i^{\intercal} \m+A \vx.
    \end{align}

	Recalling that $M = \frac{1}{4} \vx^{\intercal} \m+A \vx$ and
	using Equations~\eqref{eq:boundSUmodify-1} and~\eqref{eq:boundSUmodify-3},
    \begin{align*}
		 \sum_{i\in S} \v+e_i^{\intercal} \m+A \vx 
		 > \sum_{i\in S} \v+e_i^{\intercal} \m+A \vx + \sum_{i\in S} \v+e_i^\intercal \m+A \v+e_i.
	\end{align*}

	Thus,
	\begin{align*}
		 0 > \sum_{i \in S} \v+e_i^{\intercal} \m+A \v+e_i.  
    \end{align*}

    However, since $\m+A$ is positive semi\-definite we must have that
	$\v+e_i^{\intercal} \m+A \v+e_i \geq 0$ for all $i$ and thus 
	$\sum_{i \in S} \v+e_i^{\intercal} \m+A \v+e_i \geq 0$. This yields our
	desired contradiction.
\end{proof}

Next, let us consider how our approximation behaves when we apply 
Lemma~\ref{lemma:boundSUmodify} multiple times in a row.
Here, the issue is that we may need to apply the lemma more than
$\frac{\abs{S}}{2}$ times in a row and then a na\"{i}ve analysis would yield a cut
value of less than $M - \frac{\abs{S}}{2} \cdot \frac{2M}{\abs{S}} = 0$, i.e.,
we would not be able to obtain our desired approximation result. However, this
analysis is too pessimistic because it assumes that after each application of
the lemma, the cut decreases by a $\frac{2}{\abs{S}}$-frac\-tion with respect to the
\emph{initial} cut. Therefore, the following lemma presents a more refined
analysis, which takes into account that during each application of
Lemma~\ref{lemma:boundSUmodify}, the cut only decreases by a
$\frac{2}{\abs{S}}$-fraction with respect to the \emph{previous} cut.
A similar idea was used by Srivastav and Wolf~\cite[Lemma 1]{srivastav1998finding}
to solve the \emph{densest $k$-sub\-graph} problem. 

Intuitively, in the lemma $(S,\bar{S})$ corresponds to the cut we obtain from
the hyperplane rounding and $(T,\bar{T})$ corresponds to the $\alpha$-balanced
solution that we wish to return.
\begin{lemma}
    \label{lemma:cut-value-after-moving}
	Suppose that $\m+A$ is a symmetric, positive semi\-definite matrix with
	$\v+1^{\intercal} \m+A = \v+0^{\intercal}$ and let $\vx\in\{-1,1\}^n$.
    Let $M_0 = \frac{1}{4}\vx^{\intercal} \m+A \vx$,
	let $(S_0,\bar{S}_0)$ denote the cut induced by $\vx$ and assume that
	$\abs{S_0}\leq n/2$.
	Furthermore, let $s,t\in (0,\frac{1}{2}]$ be such that $\abs{S_0} = sn$
	and $tn$ is an integer. 
	Then there exists a set of nodes $T$ of size $\abs{T}=tn$ such that the
	cut $(T,\bar{T})$ has value at least 
	$\frac{(1-t)^2 -7(1-t)/n + 12/n^2}{(1-s)^2  + (1-s)/n} M_0$,
	if $t>s$, and value at least 
	$\frac{t^2 - t/n}{s^2 - s/n} M_0$,
	if $t<s$.
	Furthermore, $T$ can be found by repeatedly applying
	Lemma~\ref{lemma:boundSUmodify}.
\end{lemma}

The proof of Theorem~\ref{thm:unbalanced} follows from applying
Lemma~\ref{lemma:cut-value-after-moving}, where we set $t=\alpha$ and
additionally we set $S_0$ to the set $S$ from Lemma~\ref{lem:standard} which
initially has cut value at least $\frac{2}{\pi}\OPT$. Then a case distinction
for $\abs{S}>tn$ and $\abs{S}<tn$ yields the theorem.

\subsection{Greedy heuristics}
\label{sec:greedy}

Next, we discuss two greedy heuristics, 
which can be applied in two different ways. 
First they can be used to solve 
Problem~\eqref{problem:max-disagreement}
in the model with full information, 
i.e., when the graph topology and the innate opinions of all users are available.

Second, by setting $\begop_0=\v+0$,
these greedy heuristics can be used to solve Problem~\eqref{eq:our-problem}, 
and thus, be used as subroutines for the first step 
of our approach in the model with limited information. 
In other words, they can be used to substitute the SDP-based algorithm
that we presented in the previous section. 
This is particularly useful, since solving an SDP is not scalable for large graphs, 
while the greedy methods are significantly more efficient.

\sbpara{Adaptive greedy}~(\cite{chen2020network}) initializes
$\begop = \begop_0$ and performs $k$ iterations. In each iteration, for all
indices $u$ it computes how the objective function changes when setting
$\begop_u = 1$. Then it picks the index $u$ that increases
the objective function the most.

\sbpara{Non-adaptive greedy} works similarly. In a first step, it initializes
$\begop = \begop_0$ and computes for all indices $u$ the score that indicates 
how the objective function changes when setting $\begop_u = 1$. Then it orders the indices
$u_1,\dots,u_n$ such that the score is non-increasing. Now it iterates over
$i=1,\dots,n$ and for each $i$, it sets $\begop_{u_i} = 1$ if this increases the
objective function; otherwise it proceeds with $i+1$. 
The non-adaptive greedy algorithm stops after it has changed $k$~entries.

\subsection{Computational hardness}
\label{sec:hardness}

Chen and Racz~\cite{chen2020network} left it as an open problem to prove that
maximizing the disagreement of the expressed opinions is \NPhard; they studied a
version of Problem~\ref{problem:max-disagreement} in which they had an
inequality constraint $\lVert \begop - \begop_0 \rVert \leq k$ rather than the
equality constraint we study and in which the adversary could pick a solution
vector~$\begop\in[0,1]^n$. We show that this problem, as well as
Problems~\eqref{problem:max-disagreement} and~\eqref{eq:our-problem} are \NPhard.
In addition, in Corollary~\ref{cor:our-problem-np-hard}, we show that these two 
problems are \NPhard even when $k = \Omega(n)$, 
which implies that Problem~\eqref{problem:maxcut-unbalanced} is also \NPhard when $\alpha$ 
is constant. 
\begin{theorem}
\label{thm:disagreement-np-hard}
	Problem~\eqref{problem:max-disagreement} is \NPhard for
	$\m+A=\MasIdx{\DisIdx{}}$, even for $\begop_0 = \v+0$. The problem by Chen
	and Racz~\cite{chen2020network} is \NPhard, even when $\begop_0=\v+0$ and
	$k=n$.  
	Problem~\eqref{eq:our-problem} with $\m+A=\MasIdx{\DisIdx{}}$ is also \NPhard.
\end{theorem}

\begin{corollary}
	\label{cor:our-problem-np-hard}
	Problem~\eqref{eq:our-problem} with $\m+A = \MasIdx{\DisIdx{}}$ and 
	$k \in \Omega(n)$ is \NPhard. 
\end{corollary}

\section{Experimental evaluation}
\label{sec:experiments}

We empirically evaluate the methods we propose.
Due to lack of space, we only present here our results for maximizing the
disagreement. We defer our results for maximizing the polarization to 
Appendix~\ref{sec:experiments-polarization}.

Our objective is to answer the following research questions:

\begin{description}
	\item[RQ1:] Does the SDP-based method outperform the greedy methods?
	\item[RQ2:] Is there a big gap between the settings with full information and with
		limited information?
	\item[RQ3:] Which dataset parameters determine the gap between full and limited
		information?
	\item[RQ4:] How does our approach scale with respect to~$k$? 
\end{description}

Our implementations are available in on GitHub.\footnote{\, \url{https://github.com/SijingTu/KDD-23-Adversaries-With-Limited-Information}}

\sbpara{Algorithms.}
In our experiments, we consider several algorithms that work with full
information and limited information. 

First, our algorithms with \emph{full} information are as follows. 
We use the two greedy algorithms described in Section~\ref{sec:greedy};
we refer to the adaptive greedy as \AGFull\ and 
the non-adaptive greedy as \NAGFull. 
We use the suffix~\FullInfo
to indicate that they use full information.
For \AGFull we adapt the implementation of Chen and Racz~\cite{chen2020network}.\footnote{
	\label{foot:chen-code}
	\url{https://github.com/mayeechen/network-disruption}
}	

Second, we use the suffix \LimitedInfo to refer to our methods with
\emph{limited} information, which only know the network structure
(see Section~\ref{sec:algorithms}).
For picking the seed nodes, we consider the following algorithms:
\AG\ is the adaptive greedy algorithm with $\begop_0=\v+0$, 
\NAG\ is the non-adaptive greedy algorithm with $\begop_0=\v+0$, and
\SDPalgo\ is the SDP-based algorithm from Theorem~\ref{thm:unbalanced}.
\IM\ finds the seed nodes by solving the influence-maximization
problem~\cite{kempe2015maximizing} and our implementation is based on the
Martingale approach, i.e., IMM, proposed by Tang et
al.~\cite{tang2015influence}; we set the graph edge weights as in
the weighted cascade model~\cite{kempe2015maximizing}.
\Random\ randomly picks~$k$ nodes,  and 
\Deg\ picks the $k$~nodes of the highest degree.

\sbpara{Datasets.}
We present statistics for our smaller datasets in
Table~\ref{tab:disagreement-small} and for our larger datasets in
Table~\ref{tab:disagreement-large}.  
For each of the datasets, we provide the number of vertices and edges. We also
report the \emph{normalized disagreement index}
$\DisIdx{G, \begop}' = \frac{\DisIdx{G, \begop} \cdot 10^5}{\abs{E}}$, where we
normalize by the number of edges for better comparison across different datasets
and we multiply with $10^5$ because $\frac{\DisIdx{G, \begop}}{\abs{E}}$ is
typically very small. We also report average innate opinions~$\bar{\begop}_0$
and the standard deviation of the innate opinions~$\sigma(\begop_0)$.

We note that the datasets \karate, \books, \blogs, \SBM and \GplusLTWO do not 
contain ground-truth opinions. 
However, the these datasets contain ground-truth communities; thus,
we set the nodes' innate opinions by sampling from Gaussian distributions with different
parameters, based on the community membership.
More details for all
datasets are presented in Appendix~\ref{sec:experiments-datasets}.

\begin{table*}[t]
	\centering
	\caption{Results on the small datasets for the relative increase of the
		disagreement, where we set $k = 10\%\,n$. For each dataset, we have marked the
		highest value in bold and we have made the highest value for each
		setting italic.}
	\label{tab:disagreement-small}
	\resizebox{\textwidth}{!}{%
		\begin{tabular}{@{}RRRRRRRRRRRRRRRR@{}}
  \toprule
  \multirow{2}{*}{\textsf{dataset}} & 
  \multicolumn{5}{c}{dataset properties} & 
  \multicolumn{2}{c}{full information} & 
  \multicolumn{6}{c}{limited information}\\
  \cmidrule(lr){2-6} \cmidrule(lr){7-8} \cmidrule(l){9-14} 
  & \abs{V} & \abs{E} &  \DisIdx{G, {\begop_0}}' & \bar{\begop}_0 & \sigma(\begop_0) & 
  \NAGFull & \AGFull & \SDPalgo\LimitedInfo &\NAG\LimitedInfo & \AG\LimitedInfo & \Deg\LimitedInfo & \IM\LimitedInfo & \Random\LimitedInfo \\
  \midrule
\karate & 34 & 78          & 179.156 & 0.194 & 0.146 &  2.744  & \emph{2.824}    & \textbf{2.872} & 2.706 & \textbf{2.872} & 0.838 & 0.753 & 2.040\\
\books & 105 & 441         & 76.476 & 0.217 & 0.147  &  3.087  & \emph{3.457}    & \textbf{3.584} & 2.987 & 3.399 & 0.661 & 0.924 & 1.973\\
\Twitter & 548 & 3\,638    & 10.679 & 0.602 & 0.080  &  4.361  & \emph{4.468}    & \textbf{4.646} & 4.243 & 4.312 & 1.055 & 0.966 & 2.160\\
\Reddit & 556 & 8\,969     & 0.400 & 0.498 & 0.042   &  48.344 & \textbf{48.581} & \emph{48.571} & 48.356 & \emph{48.571} & 6.803 & 8.174 & 14.581\\
\SBM &
1\,000 &
74\,568 &
0.245 &
0.346 &
0.147 &
1.890 & \textbf{1.972} & \emph{1.881} & 1.444 & 1.806 & 1.123 & 1.102 & 1.301\\
\blogs & 1\,222 & 16\,717  & 18.514 & 0.205 & 0.131  & 6.518   & \textbf{6.635}  & \emph{6.555} & 6.452 & 6.462 & 0.747 & 0.708 & 2.172\\
  \bottomrule
\end{tabular}

	}
\end{table*}

\sbpara{Evaluation.} To evaluate our methods, we compare the initial
disagreement with the disagreement after the algorithms changed the innate
opinions. More concretely, let $\begop_0$ denote the initial innate opinions and
let $\begop$ denote the output of an algorithm. We report the score
$\frac{\begop^{\intercal} \m+A \begop - \begop_0^{\intercal} \m+A \begop_0}{\begop_0^{\intercal} \m+A \begop_0}$,
where $\m+A$ is one of the matrices~$\MasIdx{\DisIdx{}}$ or $\MasIdx{\PolIdx{}}$
from Section~\ref{sec:preliminaries}.
For example, if $\m+A=\MasIdx{\DisIdx{}}$ then we measure the
relative increase in disagreement compared to the initial setting.

\sbpara{Maximizing disagreement on small datasets.}
We start by studying the performance of our methods for maximizing
disagreement. We present the results on small datasets in
Table~\ref{tab:disagreement-small}, where $k=10\%\,n$. 
We consider these small datasets
as they allow us to evaluate our \SDPalgo-based algorithm, 
which does not scale to larger~graphs.

Our results in Table~\ref{tab:disagreement-small} show that for all datasets, our
limited-information algorithms, i.e., \SDPalgo\LimitedInfo, \NAG\LimitedInfo,
and \AG\LimitedInfo, perform surprisingly well.  Indeed, on all datasets these
algorithms have performance similar to the algorithms using full information.
Surprisingly, on \karate, \books, and \Twitter, \SDPalgo\LimitedInfo outperforms
the best algorithms with full information, even though only by very small
margins.  
Since \SDPalgo\LimitedInfo is the best method only on the three smallest
datasets, we believe that this exceptionally good performance 
is an artifact of the datasets being small.

Among the three algorithms with limited information, \SDPalgo\LimitedInfo
performs best on all datasets, but the gap to the two greedy algorithms
is relatively small. This answers {\bf RQ1}. 

We also observe that the SDP and greedy algorithms with limited information
achieve better results than the baselines.

Next, we note that for the \Reddit dataset, the disagreement increases by a factor of more
than~48.  A close look at the ground-truth opinions on \Reddit reveals that the
standard deviation of the innate opinions is just~0.042, and the normalized initial 
disagreement is also among the second smallest. 
These two factors make the dataset susceptible to increasing
the disagreement by a large amount.

\sbpara{Maximizing disagreement on larger datasets.}
Next, we consider the larger datasets in Table~\ref{tab:disagreement-large} with
$k=1\%\,n$. Here, we drop \SDPalgo\LimitedInfo due to scalability issues.

First, we observe that for the larger datasets, the dataset properties, such as, the normalized 
initial disagreement, the mean of the innate opinions, and the standard deviations of the 
innate opinions are similar to those of the smaller datasets.
Due to these similarities, we expect a similar gap between the full-information algorithms and 
the limited-information algorithms as in the smaller datasets.

Second, we observe that the methods with full information indeed are just slightly 
better than \NAG\LimitedInfo and \AG\LimitedInfo over all the datasets.  
The biggest gap in performance is on~\TweetSFOUR where the full-information
methods are about 40\% better. Note that both $\bar{\begop}_0=0.568$ and
$\sigma(\begop_0)=0.302$ are large for \TweetSFOUR; this is somewhat
uncharacteristic for our other datasets, which have either smaller
$\bar{\begop}_0$ or smaller $\sigma(\begop_0)$.
We also observe that the there is no clear winner between $\NAG$ and $\AG$ in
the limited information setting, which have very similar performance. In
addition, the greedy algorithms clearly outperform the baseline algorithms.
We present the running time analysis in Appendix~\ref{sec:experiments-running-time}.

\begin{table*}
	\centering
	\caption{Results on the large datasets for the relative increase of the
		disagreement, where we set $k = 1\%\,n$. For each dataset, we have marked the
		highest value in bold and we have made the highest value for each setting italic.}
	\label{tab:disagreement-large}
	\resizebox{\textwidth}{!}{%
		\begin{tabular}{@{}RRRRRRRRRRRRRRR@{}}
  \toprule
  \multirow{2}{*}{\textsf{dataset}} & 
  \multicolumn{5}{c}{dataset properties} & 
  \multicolumn{2}{c}{full information} & 
  \multicolumn{5}{c}{limited information}\\
  \cmidrule(lr){2-6} \cmidrule(lr){7-8} \cmidrule(l){9-13} 
  & \abs{V} & \abs{E} &  \DisIdx{G, {\begop_0}}' & \bar{\begop}_0 & \sigma(\begop_0) & 
  \NAGFull & \AGFull & \NAG\LimitedInfo & \AG\LimitedInfo & \Deg\LimitedInfo & \IM\LimitedInfo & \Random\LimitedInfo \\
  \midrule
\TweetSTWO  & 
1\,999 &
40\,498 &
29.024 &
0.216 &
0.258 &
\textbf{0.243} & \textbf{0.243} & 0.215 & \emph{0.219} & -0.012 & -0.011 & 0.075\\
\TweetSFOUR & 
3\,999 &
222\,268 &
10.765 &
0.568 &
0.302 &
\textbf{0.074} & \textbf{0.074} & \emph{0.053} & 0.050 & -0.009 & -0.006 & -0.001\\
\TweetMFIVE & 
4\,999 &
245\,085 &
1.256 &
0.075 &
0.077 &
\textbf{3.340} & \textbf{3.340} & \emph{3.148} & 3.041 & -0.027 & -0.024 & 0.524\\
\TweetLTWO  & 
21\,999 &
860\,884 &
5.890 &
0.106 &
0.136 &
\textbf{0.856} & \textbf{0.856} & \emph{0.844} & 0.808 & -0.024 & -0.013 & 0.222\\
\GplusLTWO  & 
22\,999 &
5\,474\,133 &
0.059 &
0.300 &
0.101 &
7.106 & \textbf{7.133} & 7.098 & \emph{7.120} & 0.015 & 0.016 & 0.502\\
  \bottomrule
\end{tabular}

	}
\end{table*}

Summarizing our results, we can answer {\bf RQ2}:
we find
that the setting with limited information is at most a factor of~$1.4$ worse than
the setting with full information. 

\sbpara{Relationship of dataset parameters and the gap between full and limited
	information.}
To understand how the dataset parameters influence the performance of our
algorithms with limited information, we perform a regression analysis and report
the results in Figure~\ref{fig:regression}. On the $y$-axis, we consider the
ratio between the best of $\NAG\LimitedInfo$ and $\AG\LimitedInfo$, which only use limited information,
and the best method with full information.
Observe that this ratio can be viewed as the gap between having full and having
limited information. On the $x$-axis, we plot the dataset parameters
$\DisIdx{G, \begop}'$, $\bar{\begop}_0$ and $\sigma(\begop_0)$.

First, we find that there is a low correlation between the ratio of
limited/full-information algorithms and the average innate
opinions~$\bar{\begop}_0$ ($R^2 = 0.17$) and the initial disagreement
$\DisIdx{G, \begop}'$ ($R^2=0.08$) in the datasets.
Second, we find that the correlation between the standard deviation of the
innate opinions~$\sigma(\begop_0)$ is moderately high ($R^2=0.62$). 

These finding align well with the intuition that if~$\sigma(\begop_0)$
is high, an adversary that only knows the graph lacks more information than
when~$\sigma(\begop_0)$ is small;
additionally, note that if $\sigma(\begop_0)$ is small, then the vector
$\epsvec$ from Theorem~\ref{thm:relationship} will have small norm and the second and the third
condition of the theorem should be satisfied on our datasets.
Similarly, it is intuitive that $\bar{\begop}_0$ should not have a large impact
on the adversary's decisions if it is not too high (here we consider datasets
		with $\bar{\begop}_0\leq 0.61$). However, in preliminary experiments
(not reported here) we also observed that if $\bar{\begop}_0$ is very large
($\bar{\begop}_0 \geq 0.8$) then the performance of the algorithms becomes much
worse.
Furthermore, it might be considered somewhat surprising that the correlation with~$\DisIdx{G, \begop}'$ is low, 
since one might intuitively expect that $\DisIdx{G, \begop}'$ and~$\sigma(\begop_0)$
should be closely related.
For this discrepancy, we note that $\DisIdx{G, \begop}'$ also involves the network structure.

Hence, we can answer {\bf RQ3}:
we find that the standard deviation of the initial opinions is the most important for
determining the gap between full and limited information, while the average
innate opinions and initial disagreement play no major role.

\begin{figure*}[t]
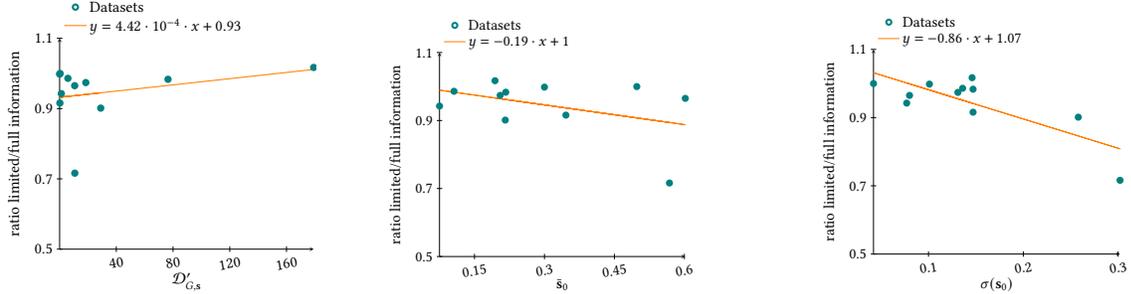

	\centering 
	\resizebox{\textwidth}{!}{%
    \begin{tabular}{ccc}
        \resizebox{0.3\textwidth}{!}{%
			\inputtikz{tikz/regression_disagreement}
		}&
        \resizebox{0.3\textwidth}{!}{%
			\inputtikz{tikz/regression_avg}
		}&
        \resizebox{0.3\textwidth}{!}{%
			\inputtikz{tikz/regression_std}
		}\\
		(a)~initial disagreement, $R^2=0.08$ &
		(b)~average innate opinion, $R^2=0.17$ &
		(c) standard deviation of opinions, $R^2=0.62$ \\
	\end{tabular}
	}
	\caption{Regression analysis of dataset parameters and the ratio between the
		best methods with full and limited information.
		}
	\label{fig:regression}
\end{figure*}

\sbpara{Dependency on~$k$.}
For $k=0.5\%\,n,1\%\,n,\dots,2.5\%\,n$, we present our results on \TweetLTWO and
\GplusLTWO in Figure~\ref{fig:scale-k}.  The figure indicates that the
disagreement grows linearly in~$k$; this behavior was also suggested by the
upper bounds of Chen and Racz~\cite{chen2020network} and Gaitonde et
al.~\cite{gaitonde2020adversarial} who considered a slightly stronger adversary.
Similar to the results in Table~\ref{tab:disagreement-large}, \AG\FullInfo is
the best method, followed by \NAG\LimitedInfo and \AG\LimitedInfo.  The ranking
of the algorithms is consistent across the different values of~$k$.
This answers {\bf RQ4}.

\begin{figure}[t]
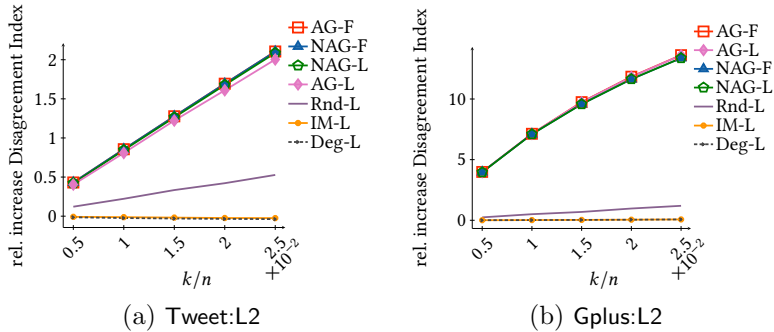

	\centering 
    \begin{tabular}{cc}
        \resizebox{0.3\columnwidth}{!}{%
			\inputtikz{tikz/twitterl2dis}
		}&
		\resizebox{0.3\columnwidth}{!}{%
			\inputtikz{tikz/gplusl2dis}
		}\\
		(a)~\TweetLTWO & (b)~\GplusLTWO \\
	\end{tabular}
	\caption{Results for the relative increase of the
		disagreement on \TweetLTWO and \GplusLTWO. Here, we vary
		$k=0.5\%\,n, 1\%\,n, \dots, 2.5\%\,n$.}
	\label{fig:scale-k}
\end{figure}

\sbpara{Additional experiments.}
In the appendix we present additional experiments.  
First, in Appendix~\ref{sec:experiments-influential}
we evaluate the algorithms for solving Problem~\eqref{eq:our-problem}.  
Second, in Appendix~\ref{sec:experiments-polarization}
we also use our algorithms to maximize the polarization in the network; we remark
that all of our results extend to this setting, including the guarantees from
Theorem~\ref{thm:unbalanced}.

\section{Conclusion}
\label{sec:conclusion}

We have studied how adversaries can sow discord in social networks, even when
they only have access to the network topology and cannot assess the opinions of
the users. We proposed a framework in which we first detect a small set of users who
are highly influential on the discord and then we change the opinions of
these users. We showed experimentally that our approach can increase the
discord significantly in practice and that it performs within a constant
factor to the greedy algorithms that have access to the full information about the
user opinions.

Our practical results demonstrate that attackers of social networks are quite
powerful, even when they can only access the network topology. From an ethical
point of view, these findings showcase the power of malicious attackers to 
sow discord and increase polarization in social networks.
However, to draw a final conclusion further study is needed, for
example because the assumption that the adversary can radicalize $k$~opinions
\emph{arbitrarily much} may be too strong. Nonetheless, the upshot is that 
by understanding attackers with limited information, one may be able to make
recommendations to policy makers regarding the data that social-network
providers can share with the public.

Furthermore, in this paper we only studied one possible definition for
disagreement that is common in the computer science
literature~\cite{chen2020network,gaitonde2020adversarial}. Klofstad et
al.~\cite{klofstad2013disagreeing} point out that in the political science
literature there are different viewpoints on how disagreement should be defined,
and that these different definitions will lead to different conclusions, with
different empirical and democratic consequences.  Understanding the connection
of our definition and the ones in political science is an interesting question.
Also, Edenberg~\cite{edenberg2021problem} argues that to solve current societal
problems like polarization, purely technical solutions, such as social media
literacy campaigns and fact checking, are not enough; instead \emph{``we must find
ways to cultivate mutual respect for our fellow citizens in order to reestablish
common moral ground for political debate.’’} 
While certainly true, such considerations and course of actions 
are out of the scope of our paper.

As we already mentioned above, in future work it will be interesting to validate
which adversary models are realistic in practice. Theoretically, it is
interesting to obtain approximation algorithms for
Problem~\eqref{problem:max-disagreement} and the problem by Chen and
Racz~\cite{chen2020network}; note that such algorithms must generalize our
result from Theorem~\ref{thm:unbalanced}, as Problem~\eqref{eq:our-problem} is a
special case of Problem~\eqref{problem:max-disagreement}.

\section*{Acknowledgements}
We are grateful to Tianyi Zhou for providing the Twitter datasets with innate
opinions. We thank Sebastian Lüderssen for pointing out a mistake in an earlier
version of this paper. This research is supported by the Academy of Finland project MLDB
(325117), the ERC Advanced Grant REBOUND (834862),
the EC H2020 RIA project SoBigData++ (871042),
and the Wallenberg AI, Autonomous Systems and Software Program (WASP)
funded by the Knut and Alice Wallenberg Foundation.
The computations were enabled by resources in project 
SNIC 2022/22-631 provided by Uppsala University at UPPMAX.

\clearpage

\bibliographystyle{ACM-Reference-Format}
\balance
\bibliography{bibclean}

\appendix
\clearpage 
\section{Omitted Experiments}
\label{sec:omit-experiments}
We present further details of our experiments. 
In Section~\ref{sec:experiments-datasets}, we present details about our
datasets.  In Section~\ref{sec:experiments-polarization}, we present how our
algorithms perform when the goal is to maximize the polarization.
In Section~\ref{sec:experiments-influential} we compare different algorithms for
finding users that are influential on the disagreement in the network
(Problem~\eqref{eq:our-problem}).
In Section~\ref{sec:experiments-stability} we discuss the stability 
of $\SDPalgo\LimitedInfo$ and $\Random\LimitedInfo$, which use randomization.
In Section~\ref{sec:experiments-running-time} we present the running time of
our algorithms.

Our algorithms are implemented in Python, except IMM~\cite{tang2015influence}
(related to our \IM\LimitedInfo algorithm) which is implemented in Julia.  We
use Mosek to solve semidefinite programs.  

\subsection{Datasets}
\label{sec:experiments-datasets}
The datasets \TweetSTWO, \TweetSFOUR, \TweetMFIVE, \TweetLTWO are sampled from
a Twitter dataset with innate opinions, which we obtained from Tianyi Zhou.
The original Twitter dataset is collected in the following way.
We start from a list of Twitter accounts who actively engage in political discussions in the US, 
which was compiled by Garimella and Weber~\cite{garimella2017long}.
Then we randomly sample a smaller subset of 50\,000 accounts. For these active
accounts, we obtained the entire list of followers, except for users with more
than 100\,000 followers for whom we got only the 100\,000 most recent
followers (users with more than 100\,000 followers account for less than 2\% of
our dataset). Based on this obtained information, we construct a graph in which the nodes
correspond to Twitter accounts and the edges correspond to the accounts'
following relationships. Then we consider only the largest connected component
in the network. To obtain the innate opinions of the nodes in the graphs, we proceed as follows.
First, we compute the political polarity score for each account using the method
proposed by Barber\'a~\citep{barbera2015birds}, which has been used widely in
the literature~\cite{brady2017emotion,boutyline2017social}.  The polarity scores
range from -2 to 2 and are computed based on following known political accounts. 
Then we re-scale them into the interval $[0,1]$.  
To create our smaller datasets, we select a seed node uniformly at random, and
run breadth-first search (BFS) from this seed node, until a given number of nodes
have been explored.

We note that for \TweetMFIVE and \TweetLTWO, the innate opinions were very large
($\bar{\begop}_0 \geq 0.85$) and thus for these datasets we flipped the innate
opinions around~0.5 (i.e., we set $\begop_0 = \v+1 - \begop_0$).  In other
words, we assume that initially most people are not on the extreme side of the
opinion spectrum.  By flipping the innate opinions, we guarantee that the
attacker can still radicalize the opinions.  We note that this has no influence
on the the initial indices for polarization and disagreement (since for 
$\m+A\in\{\MasIdx{\DisIdx{}}, \MasIdx{\PolIdx{}}\}$ it holds that
$\begop_0^\intercal \m+A \begop_0 = (\v+1 - \begop_0)^\intercal \m+A (\v+1 - \begop_0)$
which is implied by Lemma~\ref{lem:matrices}).

The dataset \GplusLTWO is sampled from the ego-Gplus dataset obtained from
SNAP~\cite{snapnets} using the same BFS-approach as above.  The innate opinions
for \GplusLTWO are drawn independently from $N(0.3, 0.1)$.  Here,
$N(\mu,\sigma)$ denotes the Gaussian distribution with mean~$\mu$ and standard
deviation~$\sigma$.

We also use the public\cref{foot:chen-code} datasets \Twitter and \Reddit from De et
al.~\cite{de2019learning}, which have previously been used by Musco et
al.~\cite{musco2018minimizing} and Chen and Racz~\cite{chen2020network}.  The
\Twitter dataset was obtained from tweets about the Delhi legislative assembly
elections of 2013 and contains ground-truth opinions. The opinions for the
\Reddit dataset were generated by Musco et al.~\cite{musco2018minimizing} using a power law
distribution.

Furthermore, we consider the datasets \karate, \books, \blogs, which we obtained
from KONECT~\cite{konect} and which do not contain ground-truth opinions.
However, these datasets contain two ground-truth communities. For the first
community, we sample the innate opinions of the users from the Gaussian
distribution $N(0.1,0.1)$, and for the second community we use $N(0.3,0.1)$.

Last, we consider graphs generated from the Stochastic Block Model.  We generate
a Stochastic Block Model graph that consists of $1000$~nodes divided equally
into $4$ communities.  The intra-community edge probability is $0.4$ and an the
inter-community edge probability is $0.1$.  The innate opinions for each
community are drawn from $N(0.2,0.1)$, $N(0.3,0.1)$, $N(0.4, 0.1)$ and
$N(0.5, 0.1)$, respectively. 

\subsection{Maximizing the polarization}
\label{sec:experiments-polarization}
Next, we use our algorithms to maximize the polarization. We remark that all of
our results extend to this setting, including the guarantees from
Theorem~\ref{thm:unbalanced}.

We report our results on larger datasets with
$k=1\%n$ in Table~\ref{tab:polarization-large}.
For each of the datasets, we provide the number of vertices and edges. We also
report the \emph{normalized polarization index} $\PolIdx{G, \begop}' =
\frac{\PolIdx{G, \begop} \cdot 10^5}{\abs{V}}$, where we normalize
by the number of vertices for better comparison across different
datasets. Further, we report average innate opinions~$\bar{\begop}_0$ and the
standard deviation of the innate opinions~$\sigma(\begop_0)$.
Finally, as before, for each algorithm we report the score
$\frac{\begop^{\intercal} \MasIdx{\PolIdx{}} \begop - \begop_0^{\intercal} \MasIdx{\PolIdx{}} \begop_0}{\begop_0^{\intercal} \MasIdx{\PolIdx{}} \begop_0}$.

\begin{table*}[t]
\centering
\caption{Results on the larger datasets for the relative increase of the
	polarization, where we set $k = 1\% n$. For each dataset, we have marked the
	highest value in bold and we have made the highest value for each setting italic.}
\label{tab:polarization-large}
\resizebox{0.9\textwidth}{!}{%
	\begin{tabular}{@{}RRRRRRRRRRRRRRR@{}}
  \toprule
  \multirow{2}{*}{\textsf{dataset}} &
  \multirow{2}{*}{$\abs{V}$} & 
  \multirow{2}{*}{$\abs{E}$} &  
  \multirow{2}{*}{$\PolIdx{G, {\begop_0}}'$} & 
  \multirow{2}{*}{$\bar{\begop}_0$} & 
  \multirow{2}{*}{$\sigma(\begop_0)$} & 
  \multicolumn{2}{c}{full information} & 
  \multicolumn{5}{c}{limited information}\\
  \cmidrule(lr){7-8} \cmidrule(l){9-13} 
  	& & & & & & \NAGFull & \AGFull & \NAG\LimitedInfo & \AG\LimitedInfo & \Deg\LimitedInfo & \IM\LimitedInfo & \Random\LimitedInfo \\
  \midrule
\TweetSTWO  & 
1\,999 &
40\,498 &
182.879 &
0.216 &
0.258 &
0.790 & \textbf{0.792} & \emph{0.692} & \emph{0.692} & -0.016 & -0.014 & 0.053\\
\TweetSFOUR & 
3\,999 &
222\,268 &
176.525 &
0.568 &
0.302 &
\textbf{0.273} & \textbf{0.273} & \emph{0.188} & 0.186 & -0.008 & -0.006 & 0.003\\
\TweetMFIVE &
4\,999 &
245\,085 &
35.005 &
0.075 &
0.077 &
5.875 & \textbf{5.876} & \emph{5.258} & \emph{5.258} & -0.035 & -0.035 & 0.160\\
\TweetLTWO  &
21\,999 &
860\,884 &
117.790 &
0.106 &
0.136 &
 1.859 & \textbf{1.865} & \emph{1.579} & 1.481 & -0.029 & -0.022 & 0.140\\
\GplusLTWO  &
22\,999 &
5\,474\,133 &
1.803 &
0.300 &
0.101 &
30.351 & \textbf{30.804} & 30.338 & \emph{30.786} & 0.002 & 0.001 & 0.441\\
\bottomrule
\end{tabular}

}
\end{table*}

We see that the results for polarization are somewhat similar to those for
disagreement: algorithms with full information are the best, but the best
algorithm that only knows the topology still achieves similar performance.

Furthermore, the best algorithms with limited information, i.e., \AG\LimitedInfo
and \NAG\LimitedInfo, consistently outperform the baselines \Deg\LimitedInfo,
\IM\LimitedInfo, and \Random\LimitedInfo; this shows that our strategy leads to
non-trivial results.  Furthermore, we observe that across all settings, simply
picking high-degree vertices, or picking nodes with large influence in the
independent cascade model are poor strategies. 

In addition, we show results for the increase of polarization in the small
datasets in Table~\ref{tab:polarization-small-addition}.  Again, the methods
with full information perform best, and again our methods generally perform
quite well and clearly outperform the baseline methods.

\begin{table*}[t]
	\centering
	\caption{Results on the small datasets for the relative increase of the
		polarization, where we set $k = 10\%\,n$. For each dataset, we have marked the
		highest value in bold and we have made the highest value for each
		setting italic.}
	\label{tab:polarization-small-addition}
	\resizebox{0.9\textwidth}{!}{%
		\begin{tabular}{@{}RRRRRRRRRRRRRRRR@{}}
  \toprule
  \multirow{2}{*}{\textsf{dataset}} & 
  \multirow{2}{*}{$\abs{V}$} & 
  \multirow{2}{*}{$\abs{E}$} &  
  \multirow{2}{*}{$\PolIdx{G, {\begop_0}}'$} & 
  \multirow{2}{*}{$\bar{\begop}_0$} & 
  \multirow{2}{*}{$\sigma(\begop_0)$} &  
  \multicolumn{2}{c}{full information} & 
  \multicolumn{6}{c}{limited information} \\
  \cmidrule(lr){7-8} \cmidrule(l){9-14} 
  & & & & & &\NAGFull & \AGFull & \SDPalgo\LimitedInfo & \NAG\LimitedInfo & \AG\LimitedInfo & \Deg\LimitedInfo &\IM\LimitedInfo & \Random\LimitedInfo \\
  \midrule
\karate & 34 & 78           & 479.556 & 0.194 & 0.146  & 2.407 & \textbf{2.412} & \emph{1.702} & 1.137 & 1.346 & 0.391 & 0.064 & 0.870\\
\books & 105 & 441          & 519.208 & 0.217 & 0.147  & 2.955 & \textbf{3.009} & \emph{2.149} & 0.282 & \emph{2.149} & -0.229 & -0.179 & 0.267\\
\Twitter & 548 & 3\,638     & 30.372 & 0.602 & 0.080    & 6.996 & \textbf{8.941} & 8.505 & 6.695 & \emph{8.526} & 1.899 & 0.248 & 1.563\\
\Reddit & 556 & 8\,969      & 0.952 & 0.498 & 0.042    & 132.834 & \textbf{133.258} & \emph{133.225} & 132.759 & \emph{133.225} & 1.741 & 2.529 & 14.524\\
\SBM & 
1000 &
74568 &
0.166 &
0.346 &
0.147 &
2.000 & \textbf{2.146} & \emph{2.014} & 1.182 & 1.867 & 0.676 & 0.660 & 0.954\\
\blogs & 1\,222 & 16\,717   & 544.049 & 0.205 & 0.131   & 12.686 & \textbf{12.789} & 12.576 & \emph{12.614} & 12.217 & 0.211 & 0.205 & 2.686\\
  \bottomrule
\end{tabular}

	}
\end{table*}

\subsection{Finding influential users}
\label{sec:experiments-influential}
In this section, we evaluate different methods for finding the influential users
which can maximize the disagreement.
More specifically, we evaluate different methods for solving 
Problem~\eqref{eq:our-problem}. 

We report our results in Figure~\ref{fig:evaluate-disagreement}.
Notice that when $\begop = \ind$ and $k=0$, the disagreement is 0. 
Thus, instead of evaluating the relative gain of the \disidx, we 
report absolute values of the \disidx. 

We observe that the baselines which pick random seed nodes, high degree nodes,
and nodes with high influence in the independent cascade model are clearly the
worst methods. Among the other methods, the \SDPalgo-based methods are typically
the best.  We observe that when $k$ is below $0.25n$, the greedy methods \AG{}
and \NAG{} often perform as well as \SDPalgo; however, when $k$ is larger than
$0.3n$, the \SDPalgo-based algorithm performs better. These observations
are in line with the analysis of Theorem~\ref{thm:unbalanced} which achieves the
best approximation ratios when $k$ is close to $0.5n$ (see also
Figure~\ref{fig:approx-all}).

\begin{figure}[t]
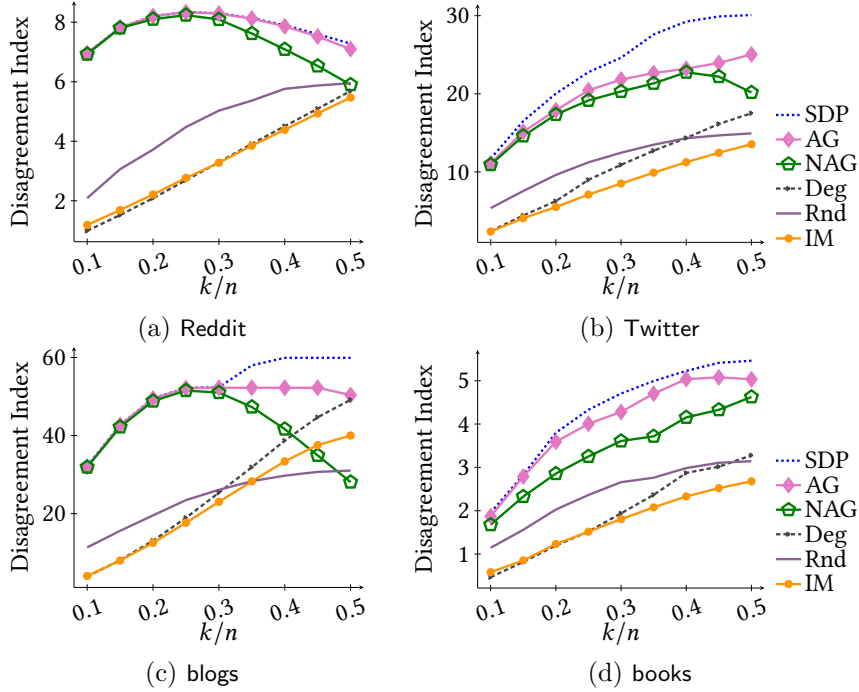

	\centering 
    \begin{tabular}{cc}
        \resizebox{0.3\columnwidth}{!}{%
			\inputtikz{tikz/reddit}
		}&
		\resizebox{0.37\columnwidth}{!}{%
			\inputtikz{tikz/real_twitter}
		}\\
		(a)~\Reddit & (b)~\Twitter \\
		\resizebox{0.3\columnwidth}{!}{%
			\inputtikz{tikz/polblogs}
		}&
		\resizebox{0.37\columnwidth}{!}{%
			\inputtikz{tikz/polbooks}
		}\\
		(c)~\blogs & (d)~\books \\
	\end{tabular}
	\caption{Results for Problem~\eqref{eq:our-problem} on different datasets.
		We varied $k = 10\%n,15\%n,...,50\%n$.}
	\label{fig:evaluate-disagreement}
\end{figure}

\subsection{Stability of randomized algorithms}
\label{sec:experiments-stability}
In this section, we study how the randomization involved in some of the
algorithms affects their results. In particular,  $\Random\LimitedInfo$ and
$\SDPalgo\LimitedInfo$ randomly select nodes and we wish to study how this
impacts their performance.  We report the \disidx and output the mean over
5~runs of the algorithms, together with error bars that indicate standard deviations.

In Figure~\ref{fig:evaluate-disagreement-std} we present the \disidx and the 
standard deviation with randomized algorithms on small datasets. 
We observe that as the number of nodes increases, the standard deviations of
different algorithms becomes relatively small (note that the largest dataset
below is \blogs).  Besides, we also observe that the outputs of
\SDPalgo\LimitedInfo are stable, with standard deviations close to 0. 

\begin{figure}[t]
	\centering 
    \begin{tabular}{cc}
        \resizebox{0.3\columnwidth}{!}{%
			\inputtikz{tikz/reddit_std}
		}&
		\resizebox{0.3\columnwidth}{!}{%
			\inputtikz{tikz/twitter_std}
		}\\
		(a)~\Reddit & (b)~\Twitter \\
		\resizebox{0.3\columnwidth}{!}{%
			\inputtikz{tikz/polblogs_std}
		}&
		\resizebox{0.3\columnwidth}{!}{%
			\inputtikz{tikz/polbooks_std}
		}\\
		(c)~\blogs & (d)~\books \\
	\end{tabular}
	\caption{Results of our randomized algorithms for
		Problem~\eqref{eq:our-problem} on different datasets.  We varied
		$k = 10\%n,15\%n,...,30\%n$ and report averages and standard deviations 
		over 5~runs of the algorithms.}
	\label{fig:evaluate-disagreement-std}
\end{figure}

In Figure~\ref{fig:evaluate-disagreement-large-std} we present results on larger graphs, 
\TweetLTWO and \GplusLTWO; here, we omit \SDPalgo\LimitedInfo due to scalability
issues.

\begin{figure}[t]
	\centering 
    \begin{tabular}{cc}
        \resizebox{0.3\columnwidth}{!}{%
			\inputtikz{tikz/twitterl2_std}
		}&
		\resizebox{0.3\columnwidth}{!}{%
			\inputtikz{tikz/gplusl2_std}
		}\\
		(a)~\TweetLTWO & (b)~\GplusLTWO \\
	\end{tabular}
	\caption{Results of $\Random\LimitedInfo$ for Problem~\eqref{eq:our-problem}
		on different datasets.  We varied $k = 0.5\%n, 1\%n,...,2.5\%n$ and
		report averages and standard deviations over 5~runs of the algorithms.}
	\label{fig:evaluate-disagreement-large-std}
\end{figure}

\subsection{Running time of algorithms}
\label{sec:experiments-running-time}
Next, we present the running time of the algorithms to maximize the disagreement
on different datasets.  Note that for full information algorithms we directly
present the running time, while for limited information algorithms we report the
running time for solving Problem~\eqref{eq:our-problem}.  This is because the
running time for setting the innate opinions to $1$ is negligible.  In addition,
since the running times of \AGFull and \AG\LimitedInfo are almost the same, we
only report the running time of \AGFull.  The same holds for \NAGFull and
\NAG\LimitedInfo.

In Figure~\ref{fig:large-dis-time}, we notice that \AGFull and \IM\LimitedInfo are 
the two most costly algorithms, but even for those  algorithm the running time
increases moderately in terms of $k$.  However, on the less dense graph
\TweetLTWO, \IM\LimitedInfo runs faster than on the denser graph \GplusLTWO,
even though they have almost the same number of vertices.  This is consistent
with the time complexity of IMM~\cite{tang2015influence}.  The graph density
does not influence the running time of the adaptive greedy algorithm \AGFull.
Interestingly, we observe that that \NAGFull is orders of magnitude faster than
\AGFull.

\begin{figure}[t]
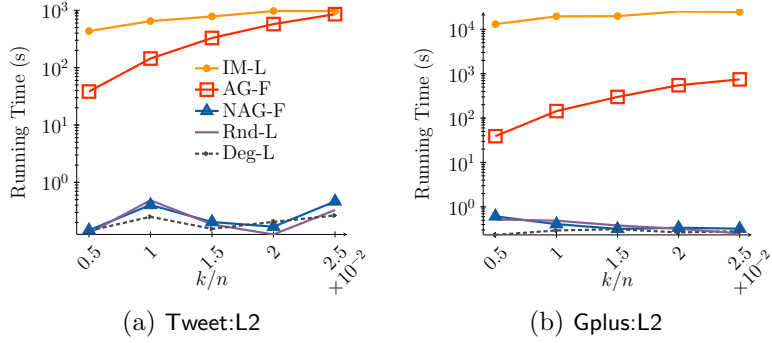

	\centering 
    \begin{tabular}{cc}
        \resizebox{0.3\columnwidth}{!}{%
			\inputtikz{tikz/twitterl2_dis_time}
		}&
		\resizebox{0.3\columnwidth}{!}{%
			\inputtikz{tikz/gplusl2_dis_time}
		}\\
		(a)~\TweetLTWO & (b)~\GplusLTWO \\
	\end{tabular}
	\caption{Running time (in seconds) of our algorithms for increasing the
			disagreement on \TweetLTWO and \GplusLTWO. Here, we vary
			$k=0.5\%n$, $1\%n$, $1.5\%n$, $2\%n$, $2.5\%n$.}
	\label{fig:large-dis-time}
\end{figure}

In Table~\ref{tab:small_running_time} we report the absolute running times (in
seconds) of our algorithms on the small datasets. We notice that
\SDPalgo\LimitedInfo is the slowest algorithm and on \blogs it is almost
30~times slower than any other algorithm. This is within our expectation, since
solving semidefinite programs is costly.  We again observe that that \NAGFull is
orders of magnitude faster than \AGFull.

\begin{table}[t]
\centering
\caption{Running time (in seconds) of our algorithms for maximizing the disagreement, where we set $k = 10\% n$.}
\label{tab:small_running_time}
\resizebox{0.6\textwidth}{!}{%
	\begin{tabular}{@{}RRRRRRRRRRRRRRRR@{}}
  \toprule
  \textsf{dataset} &
  \AGFull & \NAGFull & \SDPalgo\LimitedInfo & \IM\LimitedInfo & \Deg\LimitedInfo & \Random\LimitedInfo \\
  \midrule
 \karate
 & 0.0002
 & 0.0000
 & 0.0818
 & 0.4506
 & 0.0005
 & 0.0001
\\
\books
 & 0.0022
 & 0.0001
 & 0.4083
 & 0.5041
 & 0.0006
 & 0.0001
\\
\Twitter
 & 0.1832
 & 0.0008
 & 10.2997
 & 1.3165
 & 0.0018
 & 0.0002
\\
\Reddit
 & 0.1914
 & 0.0009
 & 9.6022
 & 2.8257
 & 0.0019
 & 0.0002
\\
\SBM 
& 1.2505
& 0.0030
& 72.8223
& 30.1188 
& 0.0036
& 0.0006
\\
\blogs
 & 3.9573
 & 0.0052
 & 171.4392
 & 5.8961
 & 0.0055
 & 0.0007
\\
  \bottomrule
\end{tabular}

}
\end{table}


\section{Omitted Proofs and Discussions}
\label{sec:omitted-proofs}
We present proofs and discussions which are not contained in the main content below.

\subsection{Rescaling Opinions}
\label{sec:scaling}

In Section~\ref{sec:preliminaries}, we mention that we consider
the opinions in the interval $[0, 1]$. 
In this appendix we prove that scaling the opinions from
an interval~$[a,b]$ to an interval~$[x,y]$ only influences the disagreement in
the network by a fixed factor of $\left(\frac{y-x}{b-a}\right)^2$. In particular,
we show that the optimizers of optimization problems are maintained under
scaling. This implies that all
\NPhardness results we derive in this paper carry over to the setting with
$[-1,1]$-opinions and our $O(1)$-approximation algorithms for $[-1,1]$-opinions
also yield $O(1)$-approximation algorithms for $[0,1]$-opinions.

Consider real numbers $a < b$ and $x < y$.
Suppose that we have innate opinions $\begop_u\in[a,b]$ and we wish to rescale
them into the interval $[x,y]$. Then we set
\begin{align*}
	\begop_u' = \begop_u \cdot \frac{y-x}{b-a}
				+ \frac{1}{2}\left( x+y - \frac{a+b}{b-a}(y-x) \right).
\end{align*}
For convenience we set $\alpha = \frac{y-x}{b-a}$ and 
$\beta = \frac{1}{2}\left( x+y - \frac{a+b}{b-a}(y-x) \right)$.
Observe that $\begop_u' = \alpha \begop_u + \beta$.
We also set $\finop_u'^{(0)} = \alpha \finop_u'^{(0)} + \beta$.

Indeed, let $f(\xi) = \alpha \xi + \beta$. Then we note that under this
transformation we have that
\begin{align*}
	f(a) &= a \frac{y-x}{b-a} + \frac{1}{2}\left( x+y - \frac{a+b}{b-a}(y-x) \right) \\
		&= \frac{1}{2} \left( 2a \frac{y-x}{b-a} + x + y - \frac{a+b}{b-a}(y-x)\right) \\
		&= \frac{1}{2} \left( x + y - \frac{b-a}{b-a} (y-x) \right) \\
		&= x,
\end{align*}
and
\begin{align*}
	f(b) &= b \frac{y-x}{b-a} + \frac{1}{2}\left( x+y - \frac{a+b}{b-a}(y-x) \right) \\
		&= \frac{1}{2} \left( 2b \frac{y-x}{b-a} + x + y - \frac{a+b}{b-a}(y-x)\right) \\
		&= \frac{1}{2} \left( x + y - \frac{a-b}{b-a} (y-x) \right) \\
		&= y.
\end{align*}
Additionally, note that $f(\xi)$ is an affine linear function. Hence, $f$ maps
$[a,b]$ bijectively into $[x,y]$.

Next, consider the expressed opinion $\finop_u'^{(t+1)}$ then by the update
rule of the FJ model and by induction we have that
\begin{align*}
	\finop_u'^{(t+1)}
	&= \frac{\begop_u' + \sum_{v \in N(u)} w_{uv} \finop_u'^{(t)}}{1 + \sum_{v\in N(u)} w_{uv}} \\
	&= \frac{\alpha \begop_u + \beta + \sum_{v \in N(u)} w_{uv} (\alpha \finop_u^{(t)} + \beta)}{1 + \sum_{v\in N(u)} w_{uv}} \\
	&= \alpha \frac{\begop_u' + \sum_{v \in N(u)} w_{uv} \finop_u'^{(t)}}{1 + \sum_{v\in N(u)} w_{uv}} + \beta \\
	&= \alpha \finop_u^{(t+1)} + \beta.
\end{align*}
In particular, in the limit we have that
$\finop' = \lim_{t\to\infty} \finop'^{(t+1)} 
	= \lim_{t\to\infty} (\alpha \finop^{(t+1)} + \beta)
	= \alpha \finop + \beta$.

Next, for the disagreement in the network we have that:
\begin{align*}
	\DisIdx{G, \begop'}
	&= \sum_{(u,v)\in E} w_{u,v} (\efinop{u}'-\efinop{v}')^2 \\
	&= \sum_{(u,v)\in E} w_{u,v} (\alpha \efinop{u} + \beta - \alpha \efinop{v} - \beta)^2 \\
	&= \alpha^2 \sum_{(u,v)\in E} w_{u,v} (\efinop{u} - \efinop{v})^2 \\
	&= \alpha^2 \DisIdx{G, \begop'}.
\end{align*}

Now we consider the mean opinion $\bar{\finop}'$:
\begin{align*}
	\bar{\finop}'
	&= \frac{1}{n} \sum_u \finop_u' \\
	&= \frac{1}{n} \sum_u (\alpha \finop_u + \beta) \\
	&= \alpha \left(\frac{1}{n} \sum_u \finop_u\right) + \beta \\
	&= \alpha \bar{\finop} + \beta.
\end{align*}
Hence, for the network polarization we obtain:
\begin{align*}
	\PolIdx{G, \begop'}
	&= \sum_{u\in V} (\finop_u' - \bar{\finop}')^2 \\
	&= \sum_{u\in V} (\alpha \finop_u + \beta - \alpha \bar{\finop} - \beta)^2 \\
	&= \alpha^2 \sum_{u\in V} (\finop_u - \bar{\finop})^2 \\
	&= \alpha^2 \PolIdx{G, \begop'}.
\end{align*}

We note that the results from above hold for all vectors $\begop\in[a,b]^n$. In
particular, this implies that if $\begop^*$ is the optimizer for an optimization
problem of the form $\max_{\begop\in[a,b]^n} \DisIdx{G, \begop'}$ then the vector
$f(\begop)\in[x,y]^n$ is the maximizer for the optimization problem 
$\max_{\begop\in[x,y]^n} \DisIdx{G, \begop'}$.

\subsection{Proof of Lemma~\ref{lem:matrices}}
We start by recalling two facts about positive semi-definite matrices. First, a
matrix $\m+A$ is positive semi-definite if $\m+A = \m+B^{\intercal} \m+C \m+B$,
where $\m+C$ is a positive semi-definite matrix. Second, $\m+A$ is positive
semi-definite if we can write it as $\m+A = \m+B^{\intercal} \m+B$.

Let us consider the matrix 
$\MasIdx{\PolIdx{}} = (\ID + \laplacian)^{-1} \left(\ID - \frac{\ind \ind^\intercal}{n}\right) (\ID + \laplacian)^{-1}$
for polarization. Observe that $\ID - \frac{\ind \ind^\intercal}{n}$ is the
Laplacian of the full graph with edge weights $1/n$ and, hence, this matrix
positive semi-definite. By our first property from above and the fact that $(\ID + \laplacian)^{-1}$ 
is symmetric, this implies that $\MasIdx{\PolIdx{}}$ is positive semidefinite.
Proving that $\MasIdx{\DisIdx{}} = (\laplacian + \ID)^{-1} \laplacian (\laplacian + \ID)^{-1}$ 
is positive semi-definite works in the same way.

Next, we observe that $(\ID+\laplacian)^{-1}$ satisfies 
$(\ID+\laplacian)^{-1} \v+1 = \v+1$ since
$(\ID+\laplacian)\v+1 = \v+1 + \laplacian\v+1 = \v+1$ and by multiplying with
$(\ID+\laplacian)^{-1}$ from both sides we obtain the claim.

Now we apply the previous observation for our matrices from the table and obtain
\begin{align*}
	\MasIdx{\PolIdx{}}\v+1
	= (\ID + \laplacian)^{-1} \left(\ID - \frac{\ind \ind^\intercal}{n}\right) (\ID + \laplacian)^{-1}\v+1
	= (\ID + \laplacian)^{-1} \left(\ID - \frac{\ind \ind^\intercal}{n}\right) \v+1
	= \v+0.
\end{align*}
And, 
\begin{align*}
	\MasIdx{\DisIdx{}} \v+1
	= (\laplacian + \ID)^{-1} \laplacian (\laplacian + \ID)^{-1} \v+1
	= (\laplacian + \ID)^{-1} \laplacian \v+1
	= \v+0.
\end{align*}

\subsection{Proof of Theorem~\ref{thm:relationship}}
	Consider the optimal solution $\begopopt$ for
	Problem~\ref{problem:max-disagreement} and let $\begopalg$ denote the
	$\beta$-approximate solution for Problem~\ref{eq:our-problem}. 
	Furthermore, set $\Deltaopt = \begopopt - \begop_0$ and
	$\Deltaalg=\begopalg-\begop_0$.

	Then we get that
	\begin{align*}
		&\frac{\begopalg^\intercal \m+A \begopalg}{\begopopt^\intercal \m+A \begopopt} \\
		&=
			\frac{(c\v+1+\epsvec+\Deltaalg)^\intercal \m+A (c\v+1+\epsvec+\Deltaalg)}
				{(c\v+1+\epsvec+\Deltaopt)^\intercal \m+A (c\v+1+\epsvec+\Deltaopt)} \\
		&=
			\frac{(\epsvec+\Deltaalg)^\intercal \m+A (\epsvec+\Deltaalg)}
				{(\epsvec+\Deltaopt)^\intercal \m+A (\epsvec+\Deltaopt)} \\
		&=
			\frac{\Deltaalg^\intercal \m+A \Deltaalg
				+ 2 \epsvec^\intercal \m+A \Deltaalg
				+ \epsvec^\intercal \m+A \epsvec}
				{\Deltaopt^\intercal \m+A \Deltaopt
				+ 2 \epsvec^\intercal \m+A \Deltaopt
				+\epsvec^\intercal \m+A \epsvec} \\
		&\geq
			\frac{\Deltaalg^\intercal \m+A \Deltaalg
				+ (1-2\gamma_1)\epsvec^\intercal \m+A \epsvec}
				{\Deltaopt^\intercal \m+A \Deltaopt
				+ 2 \epsvec^\intercal \m+A \Deltaopt
				+\epsvec^\intercal \m+A \epsvec} \\
		&\geq
			\frac{\Deltaalg^\intercal \m+A \Deltaalg
				+ (1-2\gamma_1)\epsvec^\intercal \m+A \epsvec}
				{2 (\Deltaopt^\intercal \m+A \Deltaopt
				+\epsvec^\intercal \m+A \epsvec)},
	\end{align*}
	where in the first step we used the assumption
	$\begop_0=c\v+1+\epsvec$ and the definitions of $\Deltaalg$ and $\Deltaopt$.
	In the second step we used that $\m+A \v+1 = \v+0$ by
	Lemma~\ref{lem:matrices} and that $\m+A$ is symmetric.
	In the fourth step we used Assumption~(1) and the observations that
	$\begopalg - \begop_0 = \Deltaalg$
	and
	$\begop_0^{\intercal} \m+A
		= (c\v+1 + \epsvec)^{\intercal} \m+A
		= \epsvec^{\intercal} \m+A$
	using Lemma~\ref{lem:matrices}.
	In the fifth step we used that
	$2 \epsvec^\intercal \m+A \Deltaopt \leq 
				\Deltaopt^\intercal \m+A \Deltaopt
				+\epsvec^\intercal \m+A \epsvec$.
	The last fact can be seen by letting $\m+A = D^\intercal D$ for a suitable
	matrix~$D$ (which exists since $\m+A$ is positive semi-definite by
	Lemma~\ref{lem:matrices}) and observing that
	$0 \leq || D(\Deltaopt - \epsvec) ||_2^2
		= (\Deltaopt - \epsvec)^\intercal D^\intercal D (\Deltaopt - \epsvec)
		= \Deltaopt^\intercal \m+A \Deltaopt
				- 2 \epsvec^\intercal \m+A \Deltaopt
				+\epsvec^\intercal \m+A \epsvec$;
	by rearranging terms we obtain the claimed inequality.

	Next, we let $\OPT\subseteq V$ denote the set of nodes such that
	$\begop_0(u) \neq \begopopt(u)$ and similarly we set $\ALG\subseteq V$
	to the set of nodes with $\begop_0(u) \neq \begopalg(u)$.
	Observe that since
	$\begopopt = c\v+1 + \epsvec + \Deltaopt$ and $\begopopt(u)=1$ for all
	$u\in\OPT$, we have that
	$\Deltaopt = (1-c)\v+1_{\vert \OPT} - \epsvec_{\vert \OPT}$.
	Similarly, 
	$\Deltaalg = (1-c)\v+1_{\vert \ALG} - \epsvec_{\vert \ALG}$.

	Then we get that
	\begin{align*}
		&\Deltaalg^\intercal \m+A \Deltaalg \\
		&=
			(1-c)^2 \v+1_{\vert \ALG}^\intercal \m+A \v+1_{\vert \ALG}
			- 2(1-c) \epsvec_{\vert \ALG}^\intercal \m+A \v+1_{\vert \ALG} \\
			&\quad + \epsvec_{\vert \ALG}^\intercal \m+A \epsvec_{\vert \ALG} \\
		&\geq 
			(1-c)^2 \v+1_{\vert\ALG}^\intercal \m+A \v+1_{\vert\ALG}
			- 2(1-c) \epsvec_{\vert\ALG}^\intercal \m+A \v+1_{\vert\ALG} \\
		&\geq 
			(1-c)^2 \v+1_{\vert\ALG}^\intercal \m+A \v+1_{\vert\ALG}
			- 2 (1-c) \gamma_3 \epsvec^\intercal \m+A \epsvec,
	\end{align*}
	where in the second step we used that 
	$\epsvec_{\vert\ALG} \m+A \epsvec_{\vert\ALG} \geq 0$
	and in the third step we used Assumption~(3).

	Furthermore, we obtain that
	\begin{align*}
		& \Deltaopt^\intercal \m+A \Deltaopt \\
		&=
			(1-c)^2 \v+1_{\vert\OPT}^\intercal \m+A \v+1_{\vert\OPT}
			- 2(1-c) \epsvec_{\vert\OPT}^\intercal \m+A \v+1_{\vert\OPT} \\
			&\quad + \epsvec_{\vert\OPT}^\intercal \m+A \epsvec_{\vert\OPT} \\
		&\leq
			(1-c)^2 \v+1_{\vert\OPT}^\intercal \m+A \v+1_{\vert\OPT}
			+ (2(1-c)\gamma_3+\gamma_2) \epsvec^\intercal \m+A \epsvec,
	\end{align*}
	where we used Assumptions~(3) and~(2).

	By combining our derivations from above, we obtain that
	\begin{align*}
		&\frac{\begopalg^\intercal \m+A \begopalg}{\begopopt^\intercal \m+A \begopopt} \\
		&\geq
			\frac{
				(1-c)^2 \v+1_{\vert\ALG}^\intercal \m+A \v+1_{\vert\ALG}
				+ (1 - 2\gamma_1 - 2 (1-c) \gamma_3) \epsvec^\intercal \m+A \epsvec
			}{
				2( (1-c)^2 \v+1_{\vert\OPT}^\intercal \m+A \v+1_{\vert\OPT}
					+ (1 + 2(1-c)\gamma_3+\gamma_2) \epsvec^\intercal \m+A \epsvec)
			}.
	\end{align*}
	
	Next, let $\begopoptp\in\{0,1\}^n$ denote the optimal solution for
	Problem~\ref{eq:our-problem}. 
	Furthermore, observe that $\v+1_{\vert\OPT}$ and $\v+1_{\vert\ALG}$ are
	feasible solutions for Problem~\ref{eq:our-problem}.
	Hence, we obtain
	$$ \v+1_{\vert\OPT}^\intercal\m+A\v+1_{\vert\OPT}
		\leq \begopoptp^\intercal\m+A\begopoptp.$$
	Furthermore, since in our algorithm we use an $\beta$-approximation algorithm for
	Problem~\ref{eq:our-problem} to pick the set of nodes $\ALG$, it also holds that
	$$ \v+1_{\vert\ALG}^\intercal\m+A\v+1_{\vert\ALG}
		\geq \beta \cdot \begopoptp^\intercal\m+A\begopoptp.$$
	This implies that
	\begin{align*}
		&\frac{\begopalg^\intercal \m+A \begopalg}{\begopopt^\intercal \m+A \begopopt} \\
		&\geq
			\frac{
				(1-c)^2 \beta \begopoptp^\intercal \m+A \begopoptp
				+ (1 - 2\gamma_1 - 2 (1-c) \gamma_3) \epsvec^\intercal \m+A \epsvec,
			}{
				2( (1-c)^2 \begopoptp^\intercal \m+A \begopoptp
					+ (1 + 2(1-c)\gamma_3+\gamma_2) \epsvec^\intercal \m+A \epsvec)
			}.
	\end{align*}
	
	To obtain our approximation, 
	observe that if
	$$(1-c)^2 \begopoptp^\intercal \m+A \begopoptp
				\geq (1 + 2(1-c)\gamma_3+\gamma_2) \epsvec^\intercal \m+A \epsvec)$$
	then 
	\begin{align*}
		&\frac{\begopalg^\intercal \m+A \begopalg}{\begopopt^\intercal \m+A \begopopt} \\
		&\geq
			\frac{
				(1-c)^2 \beta \begopoptp^\intercal \m+A \begopoptp
				+ (1 - 2\gamma_1 - 2 (1-c) \gamma_3) \epsvec^\intercal \m+A \epsvec,
			}{
				4(1-c)^2 \begopoptp^\intercal \m+A \begopoptp
			} \\
		&\geq \frac{\beta}{4}.
	\end{align*}
	Similarly, if 
	$$(1-c)^2 \begopoptp^\intercal \m+A \begopoptp
				< (1 + 2(1-c)\gamma_3+\gamma_2) \epsvec^\intercal \m+A \epsvec)$$
	then
	\begin{align*}
		&\frac{\begopalg^\intercal \m+A \begopalg}{\begopopt^\intercal \m+A \begopopt} \\
		&>
			\frac{
				(1-c)^2 \beta \begopoptp^\intercal \m+A \begopoptp
				+ (1 - 2\gamma_1 - 2 (1-c) \gamma_3) \epsvec^\intercal \m+A \epsvec
			}{
				4 (1 + 2(1-c)\gamma_3+\gamma_2) \epsvec^\intercal \m+A \epsvec)
			} \\
		&\geq 
			\frac{
				(1 - 2\gamma_1 - 2 (1-c) \gamma_3)
			}{
				4 (1 + 2(1-c)\gamma_3+\gamma_2)
			}.
	\end{align*}
	We conclude that the approximation ratio of our algorithm is given by
	$\frac{1}{4} \min\{\beta,
		   \frac{ 1-2\gamma_1-2(1-c)\gamma_3 }{ 1+2(1-c)\gamma_3+\gamma_2 } \}$.

\subsection{Convexity Implies Extreme Values}
\label{sec:convexity}
	Let $\m+A\in\mathbb{R}^{n\times n}$ be a symmetric matrix, let
	$\begop_0\in[-1,1]^n$ and let $k>0$ be an integer.
	Consider the following optimization problem, which generalizes
	Problems~\eqref{problem:max-disagreement},~\eqref{eq:our-problem},
	and~\eqref{problem:maxcut-unbalanced}:
	\begin{equation}
    \label{eq:convex}
		\begin{aligned}
			\max_{\begop} \quad &  \begop^{\intercal} \m+A \begop,\\
			\st \quad & \begop \in [-1,1]^n, \text{ and}\\
			& \lVert \begop - \begop_0 \rVert_0 \leq k.
		\end{aligned}
	\end{equation}   

	Now we show prove a lemma about optimal solutions of
	Problem~\eqref{eq:convex}, where we write $\v+x(i)$ to denote the $i$'th
	entry of a vector~$\v+x$.
	\begin{lemma}
	\label{lemma:max-disagreement-extreme}
	\label{lemma:convex}
		Suppose that $\m+A$ is a positive semi-definite matrix. Then there
		exists an optimal solution~$\begop$ for Problem~\ref{eq:convex} such
		that $\begop(i)\in\{-1,1\}$ for all entries~$i$ with $\begop(i)\neq\begop_0(i)$.
		In particular, if $\begop_0\in\{-1,1\}^n$ or $k=n$ then there exists an
		optimal solution~$\begop\in\{-1,1\}^n$.
	\end{lemma}
	\begin{proof}
		First, note that since $\m+A$ is positive semi-definite, the quadratic
		form $f(\begop) = \begop^{\intercal} \m+A \begop$ is convex.

		Second, consider an optimal solution $\begop$. If $\begop$ satisfies
		the property from the lemma, we are done. Otherwise, there exists at least
		one entry~$i$ such that $\begop(i)\neq\begop_0(i)$ and
		$\begop(i)\in(0,1)$. Now let $\v+t_{-1}$ denote the vector which has
		its $i$'th entry set to $-1$ and in which all other entries are the
		same as in $\begop$, i.e., 
		$\v+t_{-1}(j)=\begop(j)$ for all $j\neq i$ and $\v+t_{-1}(i)=-1$.
		Similarly, we set $\v+t_1$ to the vector with
		$\v+t_{1}(j)=\begop(j)$ for all $j\neq i$ and $\v+t_{1}(i)=1$.
		Note that $\v+t_{-1}$ and $\v+t_{1}$ are feasible solutions to
		Problem~\eqref{eq:convex}.

		Third, observe that there exists an $\alpha\in(0,1)$ such that
		$\begop = \alpha \v+t_{-1} + (1-\alpha)\v+t_1$. Now by the convexity of
		$f(\begop)$ we get that
		\begin{align*}
			f(\begop)
			&= f(\alpha \v+t_{-1} + (1-\alpha)\v+t_1) \\
			&\leq \alpha f(\v+t_{-1}) + (1-\alpha) f(\v+t_1) \\
			&\leq \max\{ f(\v+t_{-1}), f(\v+t_1)\}.
		\end{align*}
		Thus, at least one of $\v+t_{-1}$ and $\v+t_1$ achieves an objective
		function value that is at least as large as that of $\begop$. Hence, we
		can assume that the $i$'th entry of $\begop$ is from the set $\{-1,1\}$.
		Repeating the above procedure for all entries with $\begop(i)\neq\begop_0(i)$ and
		$\begop(i)\in(0,1)$ proves the first part of the lemma.

		The second part of the lemma (if $\begop_0\in\{-1,1\}^n$ or $k=n$)
		follows immediately from the first part.
	\end{proof}

\subsection{An illustration of graphs with mixed weights}
\label{sec:illustration-worst-case}
Figure~\ref{fig:cut-example} shows how the cut of a graph can be influenced by
positive \emph{and negative} weights. We use this example to show that
``badly-behaved" graphs with negative weights can make the cut value drop
significantly, whereas in graphs with only positive edges this is not the case. 
The graphs we discuss in the paper are in the class of ``well-behaved'' graphs.
\begin{figure}[h]
	\centering 
    \begin{tabular}{ccc}
        \resizebox{0.1\textwidth}{!}{%
			\includegraphics{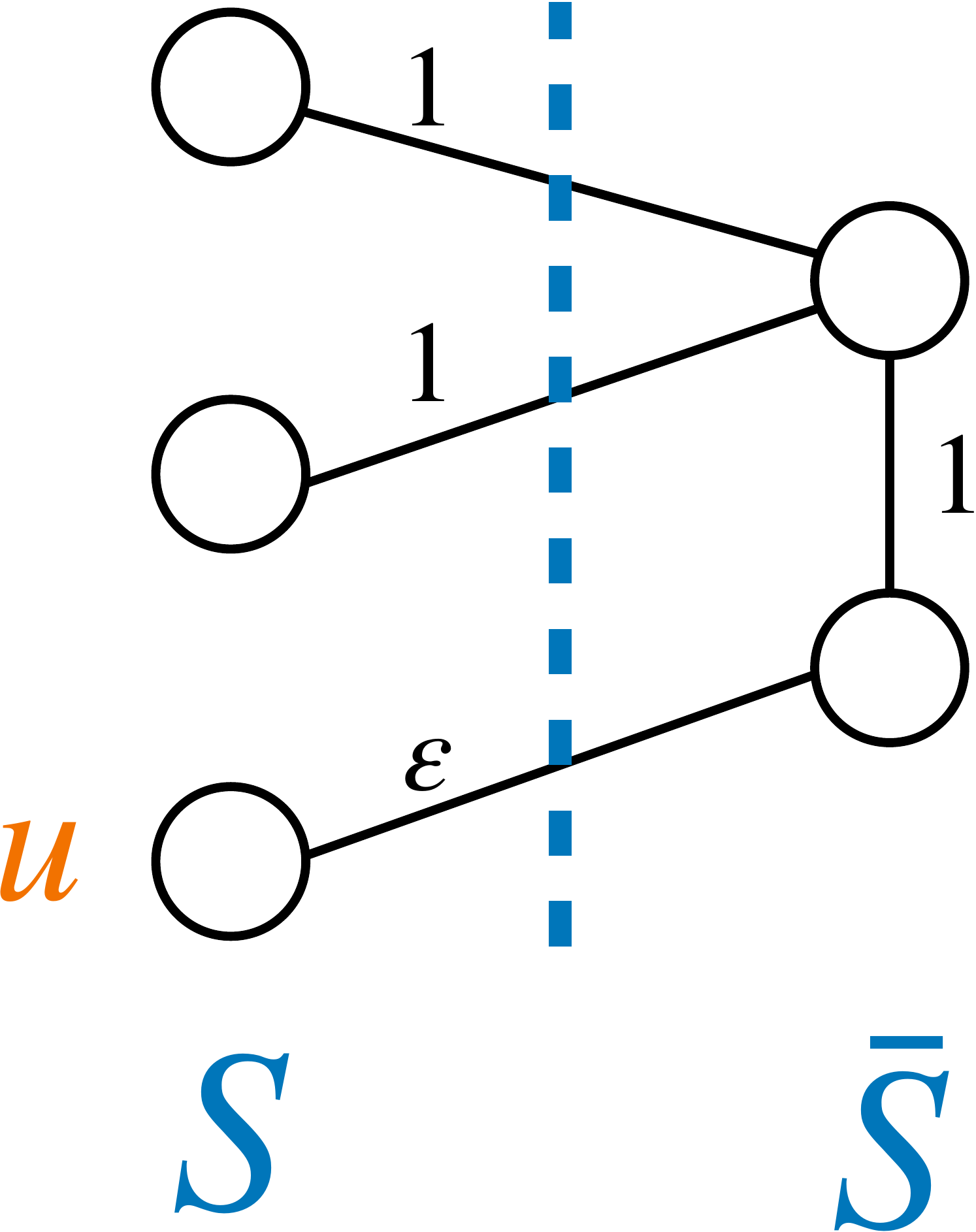}
		}&
		\resizebox{0.1\textwidth}{!}{%
			\includegraphics{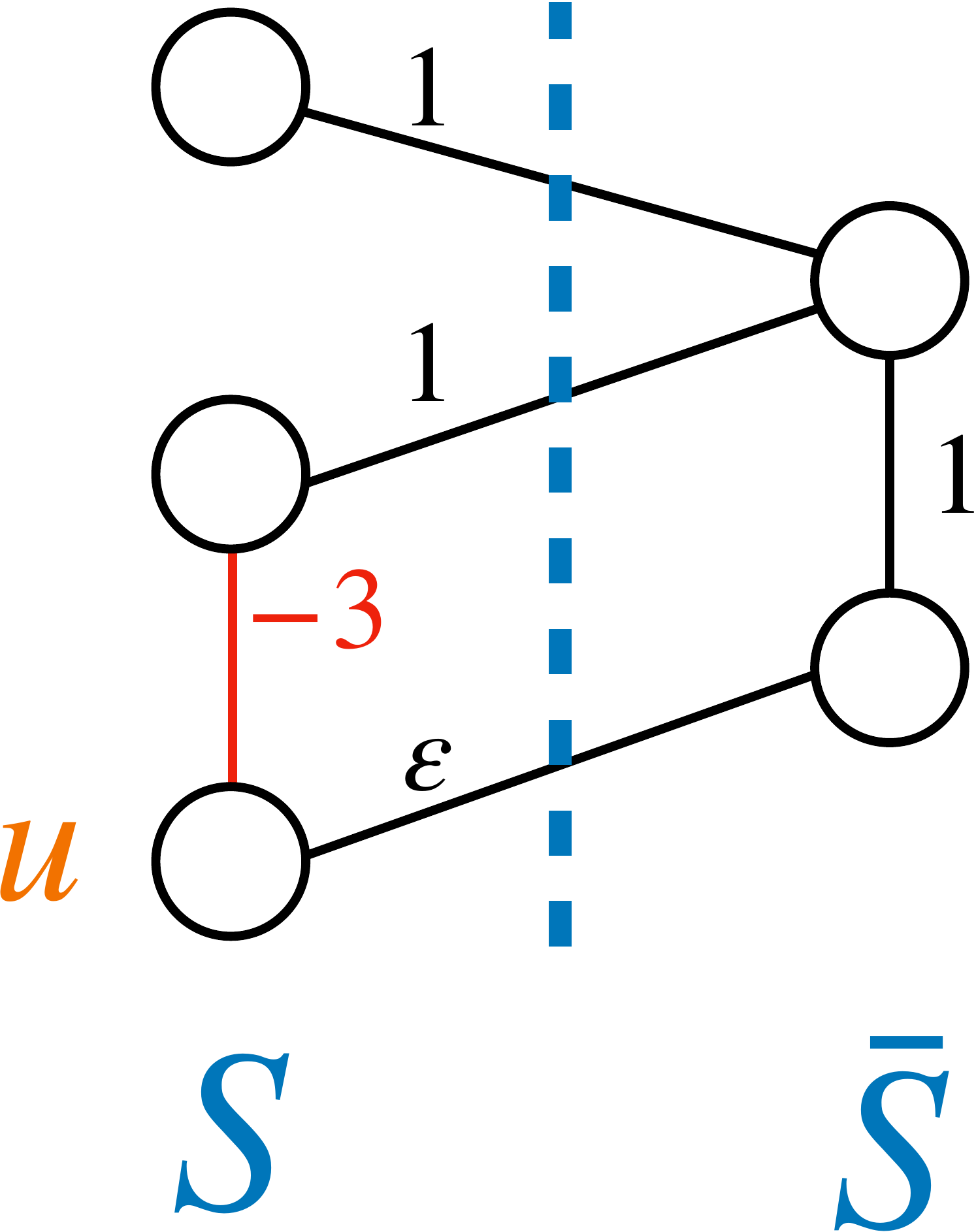}
		}&
		\resizebox{0.1\textwidth}{!}{%
			\includegraphics{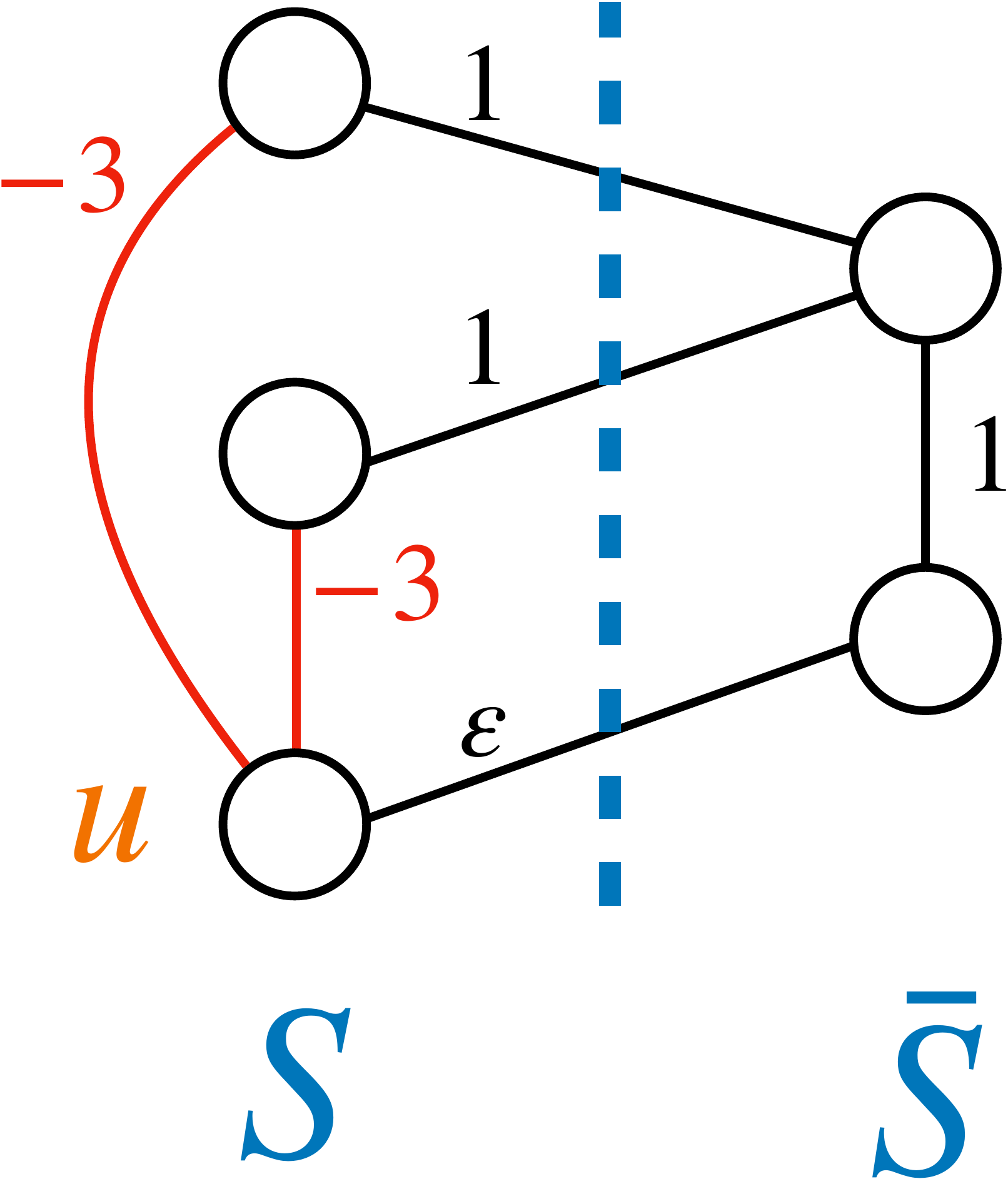}
		}\\
		(a)~Graph with
		&(b)~``Well-behaved'' 
		&(c)~``Badly-behaved''  \\
		~~positive weights &~~graph with &~~graph with \\
		&~~negative weights &~~negative weights\\
	\end{tabular}
	\caption{We visualize cuts in a graph that partition the vertices into
		sets $S$ and $\bar{S}$; edge weights are also illustrated.
		The cut shown here has value~$2+\varepsilon$.
		\emph{Figure~(a):} A graph with positive edge weights. An
				  averaging argument reveals that there must exist a vertex
				  in $S$ that we can move to $\bar{S}$ such that the cut value
				  is at least $\frac{2+\varepsilon}{\abs{S}} \geq \frac{2}{3}$.
				  Indeed, if we move $u$ from $S$ to $\bar{S}$, the cut value
				  drops only by~$\varepsilon$.
		\emph{Figure~(b):} In graphs with both positive and negative edge weights, 
				  the situation is less clear: if we move
				  $u$ from $S$ to $\bar{S}$, the cut value becomes~$-1$.
				  However, in this example we could move the top-left vertex from
				  $S$ to $\bar{S}$. Then we would still maintain a
				  positive cut value, but it would drop to $1+\varepsilon$.
				  Note that this new cut value is strictly less than what the
				  averaging argument from above revealed.
		\emph{Figure~(c):} In the worst case, it can happen that there exists no vertex
				  in~$S$, which we can move from $S$ to $\bar{S}$ without
				  obtaining a negative cut value. In our proofs, we have to show
				  that this scenario cannot occur in our setting. 
				  We prove this in Lemma~\ref{lemma:boundSUmodify}.
		}
	\label{fig:cut-example}
\end{figure}

\subsection{Proof of Theorem~\ref{thm:unbalanced}}
	We start by defining notation. 
	Let $\OPT$ denote the optimal solution for $\alpha$-{\sc Balanced-Max\-Cut},
	let $M_0$ be the objective function value obtained through randomized rounding after solving
	Problem~\eqref{problem:maxcut-unbalaced-relax}, 
	and let $M$ be the cut value for $(T, \bar{T})$ we obtain in the end. 
	Let $M^*$ be the optimal solution for Problem~\eqref{problem:maxcut-unbalaced-relax}. 

	We start with an overview of our analysis which is similar to the one by
	Frieze and Jerrum~\cite{frieze1997improved}.
	By solving the SDP and applying the hyperplane rounding enough times,
	we show that Lemma~\ref{lem:standard} implies that 
	$M_0$ is close $(1-\varepsilon) \frac{2}{\pi} \OPT$
	and simultaneously
	$\abs{S_0} \abs{\bar{S_0}} = \abs{S_0} (n - \abs{S_0})$ does not differ too
	much from $\alpha (1-\alpha) n^2$. This then implies that
	that the loss from the greedy procedure for ensuring the
	$\alpha$-balancedness constraint is not too large.
	To bound the loss from our greedy procedure for the size adjustments,
	we apply Lemma~\ref{lemma:cut-value-after-moving} with $(S_0,\bar{S}_0)$
	corresponding to the (unbalanced) solution $(S,\bar{S})$ from the hyperplane
	rounding and $(T,\bar{T})$ corresponding to the $\alpha$-balanced solution
	that we return.

	Next, we proceed with the concrete details of the proof.
	First, observe that since the objective function of the optimization problem
	is convex (see
	Section~\ref{sec:convexity}) there exists an
	optimal solution with $\begop\in\{-1,1\}^n$.  Hence, we can focus on
	solutions with $\begop\in\{-1,1\}^n$.
	
	Consider the $p$-th iteration of Algorithm~\ref{alg:SDP-based}.
	Let $X_p$ denote the cut value of $(S, \bar{S})$, and let
	$Y_p=\abs{S}\abs{\bar{S}}$. 
	Let $Z_p = \frac{X_p}{M^*} + \frac{Y_p}{N}$, where $N = \frac{n^2}{\beta}$
	and $\beta\in(0,7]$ is a parameter that depends on $\alpha$ and that we will
	pick below.  We use a similar approach as~\cite{frieze1997improved} to do the analysis.

	By Lemma~\ref{lem:standard}, $\Exp[Z_p] \geq \frac{2}{\pi} + 0.878 \alpha (1-\alpha) \cdot \beta$.
	As $X_{p} \leq M^*$ and $Y_{p} \leq \frac{n^2}{4}$, we obtain $Z_p \leq 1 + \frac{\beta}{4}$.
	We will prove that in the $\kappa$~iterations of Algorithm~\ref{alg:SDP-based}, there
	exists a $\tau$ where $Z_{\tau} = \max_{p} Z_p$ such that $Z_{\tau} \geq
	(1-\epsilon)(\frac{2}{\pi} + 0.878 \alpha (1-\alpha) \cdot \beta)$ with
	probability at least $1-\epsilon$.

	We first bound the probability of
	$Z_p \leq (1-\epsilon)\left(\frac{2}{\pi} + 0.878 \alpha (1-\alpha) \cdot \beta)\right)$
	for a single iteration~$p$ as follows:
	\begin{equation*}
		\begin{aligned}
		&\Prob{ Z_p \leq (1-\epsilon)\left(\frac{2}{\pi} + 0.878 \alpha (1-\alpha) \cdot \beta)\right) } \\
		&= \Prob{\frac{\beta}{4} + 1 - Z_p \geq \frac{\beta}{4} + 1 -
			\left((1-\epsilon)\left(\frac{2}{\pi} + 0.878 \alpha (1-\alpha) \cdot \beta
						\right)\right)} \\
		&\leq \frac{\frac{\beta}{4} + 1 - \Exp[Z_p]}{\frac{\beta}{4} + 1 - ((1-\epsilon)(\frac{2}{\pi} + 0.878 \alpha (1-\alpha) \cdot \beta))} \\
		&\leq \frac{\frac{\beta}{4} + 1 - (\frac{2}{\pi} + 0.878 \alpha (1-\alpha) \cdot \beta)}{\frac{\beta}{4} + 1 - ((1-\epsilon)( \frac{2}{\pi} + 0.878 \alpha (1-\alpha) \cdot \beta))} \\
		&= \frac{1 - c}{1 - (1-\epsilon)c},
		\end{aligned}
	\end{equation*}
	where $c = \frac{\frac{2}{\pi} + 0.878 \alpha (1-\alpha) \cdot \beta}{1 + \frac{\beta}{4}}$.
	The first equality holds as we multiply with $-1$ and add $\frac{\beta}{4} +1$ to both sides;
	the first inequality holds by Markov inequality;
	the second inequality holds by $\Exp[Z_p] \geq 0.878 \alpha (1-\alpha) \cdot \beta + \frac{2}{\pi}$.
	In the end we simplify the formula by introducing $c$. 
	Notice that since we assume $\beta > 0$, $c$ is in $(0,1)$. 
	Moreover, by Lemma~\ref{lem:helper}, we can bound $c$ between $\frac{\frac{2}{\pi}}{1}$ and $\frac{0.878 \alpha (1-\alpha) \cdot \beta}{\frac{\beta}{4}}$.
	Namely, either $c \in [{0.878 \alpha (1-\alpha) \cdot 4}, \frac{2}{\pi}]$ or $c \in [\frac{2}{\pi}, {0.878 \alpha (1-\alpha) \cdot 4}]$.

	\begin{lemma}
		\label{lem:helper}
		Let $a_1, a_2, a_3, a_4$ be positive real numbers, 
		then either $\frac{a_2}{a_4} \leq \frac{a_1 + a_2}{a_3 + a_4} \leq \frac{a_1}{a_3}$ 
		or $\frac{a_2}{a_4} \geq \frac{a_1 + a_2}{a_3 + a_4} \geq \frac{a_1}{a_3}$.
	\end{lemma}
	\begin{proof}
		We prove the lemma with two case distinctions. 

		Case 1: Assume $a_1 a_4 \geq a_2 a_3$.
		If we add $a_2 a_4$ on both sides, the formula becomes $(a_1 + a_2) a_4 \geq a_2 (a_3 + a_4)$, which implies $\frac{a_1 + a_2}{a_3 + a_4} \geq \frac{a_2}{a_4}$;
		if we add $a_1 a_3$ on both sides, the formula becomes $a_1 (a_3 + a_4) \geq (a_1 + a_2) a_3$, which implies $\frac{a_1}{a_3} \geq \frac{a_1 + a_2}{a_3 + a_4}$;
		thus $\frac{a_1}{a_3} \geq \frac{a_1 + a_2}{a_3 + a_4} \geq \frac{a_2}{a_4}$.

		Case 2: Assume $a_1 a_4 \leq a_2 a_3$.
		If we add $a_2 a_4$ on both sides, the formula becomes $(a_1 + a_2) a_4 \leq a_2 (a_3 + a_4)$, which implies $\frac{a_1 + a_2}{a_3 + a_4} \leq \frac{a_2}{a_4}$;
		if we add $a_1 a_3$ on both sides, the formula becomes $a_1 (a_3 + a_4) \leq (a_1 + a_2) a_3$, which implies $\frac{a_1}{a_3} \leq \frac{a_1 + a_2}{a_3 + a_4}$;
		thus $\frac{a_1}{a_3} \leq \frac{a_1 + a_2}{a_3 + a_4} \leq \frac{a_2}{a_4}$.
	\end{proof}
	
	Notice that as the algorithm repeats the procedure $\kappa$ times, and the
	runs of the procedure are independent from each other, the
	probability that $Z_{p} \leq (1-\epsilon)(\frac{2}{\pi} + 0.878 \alpha (1-\alpha) \cdot \beta)$
	for all $p$ is then bounded from above by
	\begin{equation*}
		\begin{aligned}
			\left(\frac{1 - c}{1 - (1-\epsilon)c}\right)^\kappa
			&= \left(1 - \frac{1}{1+ \frac{1-c}{\epsilon c}}\right)^\kappa \\
			&= \left(1 - \frac{1}{1+ \frac{1-c}{\epsilon c}}\right)^{(1+ \frac{1-c}{\epsilon c}) \frac{\epsilon c}{1 - c + \epsilon c}\kappa} \\
			&\leq \exp\left(-\frac{\epsilon c}{1 - c + \epsilon c}\kappa \right),
		\end{aligned}
	\end{equation*}
	where the inequality holds through $(1-\frac{1}{x})^x \leq \frac{1}{e}$ for any $x > 1$.

	As a result, if we choose $\kappa  \geq \frac{1 - c + \epsilon c}{\epsilon c} \log (\frac{1}{\epsilon}) \in \bigO(\frac{1}{\epsilon} \log (\frac{1}{\epsilon}))$, 
	with probability at least $1 - \epsilon$, $Z_{\tau} \geq (1-\epsilon)(\frac{2}{\pi} + 0.878 \alpha (1-\alpha) \cdot \beta)$.
	
	Now consider the (non-$\alpha$-balanced) solution $(S,\bar{S})$ from the
	$\tau$-th iteration with cut-value $M_0 := X_\tau$. We set
	$s=\frac{\abs{S}}{n}$ and $t=\alpha$ and apply
	Lemma~\ref{lemma:cut-value-after-moving} to obtain a solution that satisfies
	the $\alpha$-balancedness constraint.
	Suppose that in the $\tau$-th run, $X_{\tau} = \lambda M^*$ for suitable
	$\lambda$.  Then it follows that $Y_{\tau} \geq ((1-\epsilon)(0.878 \alpha (1-\alpha) \cdot \beta + \frac{2}{\pi}) - \lambda) N$.
	By replacing $Y_{\tau}$ with $Y_{\tau} = n^2(1-s)s$, we obtain
	$\lambda \geq (1 - \epsilon)(\frac{\pi}{2} + 0.878 \alpha (1-\alpha) \cdot \beta) - \beta (1-s)s$. 

	We distinguish the two cases $0 < s \leq t \leq 0.5$ and $0< t < s \leq 0.5$.
	
	\textbf{Case 1:} $0 < s < t \leq 0.5$. By
	Lemma~\ref{lemma:cut-value-after-moving} we have that
	$M \geq \frac{(1-t)^2 -7(1-t)/n + 12/n^2}{(1-s)^2  + (1-s)/n} M_0$.
	Now notice that $\lim_{n \rightarrow \infty} \frac{(1-t)^2 -7(1-t)/n + 12/n^2}{(1-s)^2  + (1-s)/n} = \frac{(1-t)^2}{(1-s)^2}$. 
	Which indicates that for any $\epsilon' > 0$, 
	there exists a constant $C$, such that for any $n \geq C$,
	$M \geq \left(\frac{(1-t)^2}{(1-s)^2} - \epsilon'\right) M_{0}$. 
	Using the above analysis, and setting $t = \alpha$, 
	we obtain that $M \geq \lambda \frac{(1-\alpha)^2}{(1-s)^2} \OPT
	\geq [\frac{\pi}{2} + 0.878 \alpha (1-\alpha) \cdot \beta - \beta (1-s)s - \varepsilon'']\frac{(1-\alpha)^2}{(1-s)^2} \OPT$. 

	\textbf{Case 2:} $0 < t < s  \leq 0.5$. By Lemma~\ref{lemma:cut-value-after-moving}, 
	we obtain $M \geq \frac{t^2 - t/n}{s^2 - s/n} M_0$.
	Notice that $\lim_{n \rightarrow \infty} \frac{t^2 - t/n}{s^2 - s/n} = \frac{t^2}{s^2}$,
	which indicates that for any $\epsilon' > 0$, there exists a constant $C$,
	such that for any $n \geq C$, $M \geq \left(\frac{t^2}{s^2} - \epsilon'\right) M_{0}$. 
	Using the above analysis, and setting $t = \alpha$,
	we obtain that $M \geq \lambda \frac{\alpha^2}{s^2} \OPT
	= [\frac{\pi}{2} + 0.878 \alpha (1-\alpha) \cdot \beta - \beta (1-s)s - \varepsilon''']\frac{\alpha^2}{s^2} \OPT$. 
	
	Now we combine the two cases together, i.e., we take the minimum of the two
	solutions given any $\alpha \leq 0.5$.  To do so, we numerically solve the
	following problem for any $\alpha\leq 0.5$, 
	\begin{align*}
	\max_{0<\beta <7}
	\quad
	\min_{0< s_1 < \alpha \leq 0.5, \,\, 0<\alpha \leq s_2 \leq 0.5}
	&\left\{
		\left( \frac{\pi}{2} + 0.878 \alpha (1-\alpha) \cdot \beta - \beta
				(1-s_1)s_1\right)\frac{(1-\alpha)^2}{(1-s_1)^2}, \right. \\
		&\quad \left. \left(\frac{\pi}{2} + 0.878 \alpha (1-\alpha) \cdot \beta - \beta (1-s_2)s_2\right)\frac{\alpha^2}{s_2^2}
	\right\}.
	\end{align*}
	Notice that we only consider values for $\beta$ in a small domain $(0, 7)$
	since the running time of our algorithm depends on $\beta$.  We present the
	approximation ratio for given $\alpha$ and the choice of $\beta$ in
	Table~\ref{tab:selected_ratio}. 

	\begin{table*}[t]
	\centering
	\caption{The approximation ratio of Algorithm~\ref{alg:SDP-based} for some
		values of $\alpha$, and the corresponding choice of $\beta$.}
	\label{tab:selected_ratio}
	\resizebox{0.3\textwidth}{!}{%
		\begin{tabular}{@{}RRRRRRRRRRRR@{}}
\toprule
{\alpha} & \textsf{Approximation Ratio} & {\beta} \\
\midrule
0.01 & 0.0002 & 0.01\\
0.05 & 0.0063 & 0.01\\
0.09 & 0.0205 & 0.01\\
0.13 & 0.0429 & 0.01\\
0.17 & 0.0734 & 0.01\\
0.21 & 0.1121 & 0.01\\
0.25 & 0.1589 & 0.01\\
0.29 & 0.2139 & 0.01\\
0.33 & 0.2770 & 0.01\\
0.37 & 0.3268 & 0.877\\
0.41 & 0.3754 & 2.081\\
0.45 & 0.4076 & 4.075\\
0.49 & 0.3635 & 5.390\\
0.525 & 0.3837 & 5.341\\
0.565 & 0.4017 & 3.088\\
0.605 & 0.3570 & 1.598\\
0.645 & 0.3090 & 0.478\\
0.685 & 0.2524 & 0.01\\
0.725 & 0.1923 & 0.01\\
0.765 & 0.1404 & 0.01\\
0.805 & 0.0966 & 0.01\\
0.845 & 0.0610 & 0.01\\
0.885 & 0.0335 & 0.01\\
0.925 & 0.0142 & 0.01\\
0.965 & 0.0031 & 0.01\\
\bottomrule
\end{tabular}
	}
	\end{table*}

	Note that similar analysis holds when $\alpha > 0.5$, since essentially 
	in this case, our greedy procedure starts from a $S_0$ where $\abs{S_0} > \frac{1}{2}n$ (otherwise, 
	we take $\bar{S}_0$ as $S_0$). 
	The approximation ratio holds a symmetric property. 

	We plot the approximation ratio of our algorithm for different values of
	$\alpha$ in Figure~\ref{fig:approx-all}.

\subsubsection{Proof of Lemma~\ref{lem:standard}}
\begin{proof} 
	The analysis by Williamson and 
	Shmoys~{\cite[Theorem 6.16]{williamson2011design}} shows that the expected
	cut of $(S, \bar{S})$ is not less than $\frac{2}{\pi} M^*$, 
	where $M^*$ denotes the optimal solution for Problem~\ref{problem:maxcut-unbalaced-relax}. 
	Their analysis assumes that there is no $\ell_0$ constraint; however, 
	their proof still holds in our setting, since the $\ell_0$ constraint and its relaxation 
	do not influence the analysis and since
	Problem~\eqref{problem:maxcut-unbalaced-relax} is a relaxation of
	Problem~\eqref{eq:our-problem}.
	To prove that $\Exp[\abs{S}\abs{\bar{S}}] \geq 0.878 \cdot \alpha (1- \alpha) n^2$, 
	we use the same method as {\cite[Section 3]{frieze1997improved}}.

	For the sake of completeness, we now present the details of the analysis.
	
	\begin{lemma}[Lemma 6.12~\cite{williamson2011design}]
		\label{lemma:maxbisectionrelax}
		$\Exp[\x_i\x_j] = \frac{2}{\pi} \arcsin(\v+v_i \cdot \v+v_j)$. 
	\end{lemma}
	\begin{corollary}
		\label{corollary:maxbisectionrelcut}
		$\frac{1}{2}\Exp[1 - \x_i\x_j] = \frac{1}{\pi}\arccos(\v+v_i \cdot \v+v_j)$. 
	\end{corollary}
	\begin{proof}
		This is because $\arccos(x) = \frac{\pi}{2} - \arcsin(x)$ for any $x$. 
	\end{proof}

	\begin{lemma}[Corollary 6.15~\cite{williamson2011design}]
		\label{lemma:maxbisectionrelaxcorollary}
		If $\m+X \succcurlyeq \m+0$, $X_{ij} \leq 1$ for all $i, j$ and 
		$\m+Z = (Z_{ij})$ such that $Z_{ij} = \arcsin(X_{ij}) - X_{ij}$, 
		then $\m+Z \succcurlyeq \m+0$. 
	\end{lemma}

	We are ready to prove the first part of Lemma~\ref{lem:standard}.
	Note the expected cut of $(S, \bar{S})$ can be formulated as $\Exp[\sum_{i, j}\frac{1}{4}A_{ij} \x_i\x_j]$. 
	Then we have:
	\begin{displaymath}
		\begin{aligned}
			\Exp[\sum_{i, j}\frac{1}{4}A_{ij} \x_i\x_j]
			&= \sum_{i, j}\frac{1}{4}A_{ij} \Exp[\x_i\x_j] \\
			&= \frac{2}{\pi} \frac{1}{4} \sum_{i,j}A_{ij}\arcsin(\v+v_i \cdot \v+v_j) \\
			&\geq \frac{2}{\pi} \frac{1}{4} \sum_{i,j}A_{ij}\v+v_i \cdot \v+v_j \\
			&= \frac{2}{\pi} M^* \geq \frac{2}{\pi} \OPT,
		\end{aligned}
	\end{displaymath}
	where the second step is based on Lemma~\ref{lemma:maxbisectionrelax}, 
	and the third step is based on Lemma~\ref{lemma:maxbisectionrelaxcorollary}.

	To prove the second part of Lemma~\ref{lem:standard}, we bound the imbalance
	of the partition that we get from the random rounding procedure.
	We will show that $\abs{S}\abs{\bar{S}}$ does not deviate from $\alpha (1-\alpha) n^2$ too much.

	\begin{lemma}[Lemma 6.8~\cite{williamson2011design}]
		\label{lemma:maxbisectionrelcutbounc}
		For $x \in [-1, 1]$, $\frac{1}{\pi} \arccos(x) \geq 0.878 \cdot \frac{1}{2} (1 - x)$.  
	\end{lemma}

	Now we can calculate $\Exp[\abs{S}\abs{\bar{S}}]$:
	\begin{displaymath}
		\begin{aligned}
		\Exp[\abs{S}\abs{\bar{S}}]
		&= \sum_{i<j}\frac{1}{2} \Exp[1 - \x_i\x_j] \\
		&= \sum_{i<j} \frac{1}{\pi} \arccos(\v+v_i \cdot \v+v_j) \\
	   	&\geq 0.878 \sum_{i<j}\frac{1}{2} (1 - \v+v_i \cdot \v+v_j)  \\
		&= 0.878 \left(\frac{n^2 - n}{4} - \frac{1}{2}\sum_{i<j}\v+v_i^{\intercal}\v+v_j\right) \\
		&= 0.878 \left(\frac{n^2 - n}{4} - \frac{1}{4}(n^2(1-2\alpha)^2 - n) \right) \\
		&= 0.878\cdot \alpha (1-\alpha) n^2,
		\end{aligned}
	\end{displaymath}
	where the second step is based on
	Corollary~\ref{corollary:maxbisectionrelcut}, the third step is based on
	Lemma~\ref{lemma:maxbisectionrelcutbounc}, and the fourth step is based on
	the fact that Problem~\ref{problem:maxcut-unbalaced-relax} is a semidefinite
	relaxation of Problem~\ref{problem:maxcut-unbalanced} and the fifth step
	uses the constraint from the SDP relaxation.

\end{proof}

\subsubsection{Proof of Lemma~\ref{lemma:cut-value-after-moving}}
\begin{proof}
	We repeatedly apply Lemma~\ref{lemma:boundSUmodify} until our set has size~$tn$.
	More concretely, if $\abs{S_0}>tn$, we use Lemma~\ref{lemma:boundSUmodify} to remove vertices from
	$S_0$ one by one, 
	and if $\abs{S_0}< tn$, we use Lemma~\ref{lemma:boundSUmodify} to remove vertices from $\bar{S}_0$ one by one. 
	We let $S_i$ be the set of vertices after removing
	$i$~vertices. When this process terminates we denote the resulting set by~$T$.
	Let $M_i$ denote the value of the cut given by the partition
	$(S_i,\bar{S_i})$. 
	We will distinguish the two cases $s < t$ and $s > t$. 
	
	First, suppose $s > t$. Observe that by Lemma~\ref{lemma:boundSUmodify},
	\begin{align*}
		M_i \geq M_{i-1} - \frac{2 M_{i-1}}{\abs{S_{i-1}}}
			= \left(1- \frac{2}{\abs{S_{i-1}}}\right) M_{i-1},
	\end{align*} 
	for all $i\geq 1$. Note that for $k=\abs{s-t}n$, it is $T=S_k$, and thus $M_k$
	is the value of the cut $(T,\bar{T})$.  
	By recursively applying the above inequality, we obtain that
	\begin{align}
	\label{eq:Mk}
		M_k \geq M_0 \prod_{i=1}^k \left(1- \frac{2}{\abs{S_{i-1}}}\right).
	\end{align}

	Now let us consider the term $\prod_{i=1}^k \left(1- \frac{2}{\abs{S_{i-1}}}\right)$.
	We are removing vertices from the $S_i$, and thus $\abs{S_i}=\abs{S_{i-2}}-2$. 
	Hence,
	\begin{align*}
		\prod_{i=1}^k \left(1- \frac{2}{\abs{S_{i-1}}}\right)
		    &=\prod_{i=0}^{k-1} \frac{\abs{S_{i}}-2}{\abs{S_{i}}} \\
			&= \frac{\abs{S_{0}}-2}{\abs{S_{0}}} \cdot
					\frac{\abs{S_{1}}-2}{\abs{S_{1}}} \cdot
					\frac{\abs{S_{2}}-2}{\abs{S_{2}}}
					\cdots
					\frac{\abs{S_{k-1}}-2}{\abs{S_{k-1}}} \\
			&= \frac{\abs{S_{k-2}}-2}{\abs{S_{0}}} \cdot
					\frac{\abs{S_{k-1}}-2}{\abs{S_{1}}} \\
			&= \frac{tn(tn-1)}{sn(sn-1)}
			= \frac{t^2n^2 - tn}{s^2n^2 - sn} = \frac{t^2 - t/n}{s^2 - s/n}.
	\end{align*}

	Combining the result for $\prod_{i=1}^k \left(1- \frac{2}{\abs{S_{i-1}}}\right)$
	with Equation~\eqref{eq:Mk} we obtain the claim of the lemma for $s>t$. 

	Second, consider the case $s<t$. 
	We apply Lemma~\ref{lemma:boundSUmodify} to we remove vertices from 
	$\bar{S}_i$. Since $\abs{\bar{S}_i} = n - \abs{S_i}$, we have that 
	\begin{align*}
		M_i \geq M_{i-1} - \frac{2 M_{i-1}}{n - \abs{S_{i-1}}}
			= \left(1- \frac{2}{n - \abs{S_{i-1}}}\right) M_{i-1},
	\end{align*} 
	for all $i\geq 1$. Let $k=\abs{s-t}n$, so  $T=S_k$, and $M_k$
	is the value of the cut $(T,\bar{T})$.  
	By recursively applying this inequality,
	\begin{align}
	\label{eq:Mk-2}
		M_k \geq M_0 \prod_{i=1}^k \left(1- \frac{2}{n - \abs{S_{i-1}}}\right).
	\end{align}

	Removing vertices from $\bar{S}_i$ is equivalent with the procedure of 
	adding vertices to the $S_i$ and thus $\abs{S_{i}} = \abs{S_{i-2}}+2$. 
	Hence,
	\begin{align*}
		\prod_{i=1}^k \left(1\,- \frac{2}{n - \abs{S_{i-1}}}\right)
			&= \prod_{i=0}^{k-1} \frac{n - \abs{S_{i}}-2}{n - \abs{S_{i}}} \\
			&= \frac{n - \abs{S_{0}}-2}{n - \abs{S_{0}}} \cdot
					\frac{n - \abs{S_{1}}-2}{n - \abs{S_{1}}}
					\cdots
					\frac{n - \abs{S_{k-1}}-2}{n - \abs{S_{k-1}}} \\
			&= \frac{n - \abs{S_{k-2}}-2}{n - \abs{S_{0}}} \cdot
					\frac{n-\abs{S_{k-1}}-2}{n -\abs{S_{1}}} \\
			&= \frac{(n - tn-4)(n-tn-3)}{(n-sn)(n-sn+1)} \\
			&= \frac{(1-t)^2n^2 -7(1-t)n + 12}{(1-s)^2 n^2 + (1-s)n} \\
			&= \frac{(1-t)^2 -7(1-t)/n + 12/n^2}{(1-s)^2  + (1-s)/n}.
	\end{align*}

	Combining the result for $\prod_{i=1}^k \left(1- \frac{2}{\abs{S_{i-1}}}\right)$
	with Equation~\eqref{eq:Mk}, 
	and the result for $\prod_{i=1}^k \left(1- \frac{2}{n - \abs{S_{i-1}}}\right)$
	with Equation~\eqref{eq:Mk-2},
	we obtain the claim of the lemma.
\end{proof}

\subsection{Proof of Theorem~\ref{thm:disagreement-np-hard}}
We dedicate the rest of this section to the proof of this theorem.
For technical reasons, it will be
convenient for us to consider opinion vectors $\begop,\begop_0\in[-1,1]^n$. Our
hardness results still hold for opinions vectors $\begop,\begop_0\in[0,1]^n$ by
the results from Section~\ref{sec:scaling} which shows that we only lose a fixed
constant factor and that maximizers of optimization problems are the same under
a simple bijective transformation.

We also note that in the following, we prove that Problem~\eqref{problem:max-disagreement-without-constraint} is \NPhard. 
Notice that Problem~\eqref{problem:max-disagreement-without-constraint} is a variant of 
Problem~\eqref{problem:max-disagreement} by removing cardinality constraint, setting $\m+A = \MasIdx{\DisIdx{}}$, 
and setting $\begop_0 = -\v+1$:
\begin{equation}
\label{problem:max-disagreement-without-constraint}
	\begin{aligned}
		\max_{\begop} \quad &  \begop^{\intercal} \MasIdx{\DisIdx{}} \begop,\\
		\st \quad & \begop(u) \in \{-1, 1\} \text{ for all } u\in V.
	\end{aligned}
\end{equation}  

We thus answer the question by 
Chen and Racz~\cite{chen2020network}, by setting the $k$ in their problem to be $n$.  
We remark the hardness of Problem~\eqref{problem:max-disagreement-without-constraint} 
implies hardness for Problem~\eqref{problem:max-disagreement} as follows: 
If we can solve the Problem~\eqref{problem:max-disagreement}, i.e. the problem
with an equality constraint $\lVert\begop-\begop_0\rVert=k'$, then we can solve
the problem with $\begop_0 = \mathbf{0}$ and for all values $k'=1,\dots,n$, and 
take the maximum over all answers.
This gives us an optimal solution for Problem~\eqref{problem:max-disagreement-without-constraint}. 
Since there are only $n$~choices
for~$k'$ and since we show hardness for Problem~\eqref{problem:max-disagreement-without-constraint},
we obtain hardness for Problem~\eqref{problem:max-disagreement}.

We first prove the hardness of two
auxiliary problems and then give the proof of the theorem. We start by introducing
Problem~\eqref{problem:max-cut-variant}, which is a variant of {\sc Max\-Cut}; 
compared to classic {\sc Max\-Cut}, we scale the objective function by factor~4 and consider the
constraints $\begop\in[-1,1]^n$ rather than $\begop\in\{-1,1\}^n$. We show that
the problem is \NPhard.

\begin{problem}
    \label{problem:max-cut-variant}
	Let $G=(V,E,w)$ be an undirected weighted graph with integer edge
	weights and let $\laplacian$ be the Laplacian of $G$. We want to solve the
	following problem:
\begin{equation*}
    \begin{aligned}
        \max_{\begop} \quad & \begop^{\intercal} \laplacian \begop\\
        \st \quad & \begop \in [-1,1]^n.
    \end{aligned}
\end{equation*}   
\end{problem}

\begin{lemma}[{\cite{GAREY1976237}}]
\label{problem:two-hard}
	Problem~\eqref{problem:max-cut-variant} is \NPhard, even in unweighted graphs.
\end{lemma}

Next, in Problem~\eqref{problem:cut-middle-problem} we consider a version
of Problem~\eqref{problem:max-disagreement-without-constraint}.
\begin{problem}
    \label{problem:cut-middle-problem}
    Let $Q$ be an integer. 
	Let $G=(V,E,w)$ be an undirected weighted graph with integer edge
	weights and let $\laplacian$ be the Laplacian of $G$.
    We want to solve the following problem:
\begin{equation*}
    \begin{aligned}
        \max_{\begop} \quad & \finop^{\intercal} \laplacian \finop\\
        \st \quad & \begop \in [-1,1]^n, \text{ and}\\
        & \finop = (\ID + \frac{1}{Q}\laplacian)^{-1} \begop.
    \end{aligned}
\end{equation*}   
\end{problem}
Note that by substituting the constraint
$\finop=(\ID + \frac{1}{Q}\laplacian)^{-1} \begop$ in the objective function, we
obtain our original objective function from
Problem~\eqref{problem:max-disagreement} and Problem~\eqref{problem:max-disagreement-without-constraint} 
for $Q=1$.

Next, we show that Problem~\eqref{problem:cut-middle-problem} is \NPhard.
Here, our proof strategy is as follows.
Consider an instance of Problem~\eqref{problem:max-cut-variant}. Then,
intuitively, for large $Q$ it should hold that 
$\finop = (\ID + \frac{1}{Q}\laplacian)^{-1} \begop
	\approx (\ID + 0)^{-1} \begop = \begop$
and in this case the optimal solutions of
Problems~\eqref{problem:cut-middle-problem} and~\eqref{problem:max-cut-variant} 
should be almost identical.
Indeed, we will be able to show that both problems have the same maximizer and 
that their objective function values are almost identical (up to rounding).
This implies the hardness of Problem~\eqref{problem:cut-middle-problem} via
Lemma~\ref{problem:two-hard}. 
We summarize our result in the following lemma, where we use the following
notation: 
$\OPTTWO(V, E, w)$ and
$\OPTTHREE(V, E, w)$ denote the optimal objective values of
Problems~\eqref{problem:max-cut-variant} and~\eqref{problem:cut-middle-problem}, respectively,
and $[\cdot]$ denotes rounding to the nearest integer.

\begin{lemma}
    \label{lemma:cut-middle-cut-variat}
	There exists an integer $Q$ such that $\OPTTWO = [\OPTTHREE]$.
	Moreover, the optimal solution~$\begop^*$ for Problem~\eqref{problem:cut-middle-problem}
	is also the optimal solution for Problem~\eqref{problem:max-cut-variant}. 
    This implies that Problem~\eqref{problem:cut-middle-problem} is \NPhard. 
\end{lemma} 

\subsubsection{Proof of Lemma~\ref{problem:two-hard}}
    Since $\frac{1}{4} \begop^{\intercal} \laplacian \begop = \frac{1}{2} \sum_{i < j}w_{ij}(1 - \ebegop{i}\ebegop{j})$ for $\begop \in \{-1, 1\}^n$, we 
    can formulate the MaxCut problem~\cite{goemans1995improved} as
	\begin{equation*}
		\begin{aligned}
			\max_{\begop} \quad & \frac{1}{4}\begop^\intercal \laplacian \begop\\
			\st \quad & \begop \in \{-1,1\}^n
		\end{aligned}
	\end{equation*}   
	Thus, $\OPT_{MaxCut}(V, E, w) = \frac{1}{4} \OPTTWO(V, E, w)$. 
	Given that MaxCut is \NPhard in unweighted graphs, we obtain that
	Problem~\ref{problem:max-cut-variant} is \NPhard in unweighted graphs.

\subsubsection{Proof of Lemma~\ref{lemma:cut-middle-cut-variat}}
    The proof has four steps. 
	(1)~We argue that the optimal solution $\begop$ for
	Problem~\ref{problem:cut-middle-problem} must be from the set~$\{-1,1\}^n$.
	(2)~We derive a useful characterization of the entries $\finop_i$ of
	$\finop$.
	(3)~We exploit the previous characterization by showing that there exists a
	large enough $Q$ such that $[\finop^{\intercal} \laplacian \finop] = \begop \laplacian \begop$. 
	(4)~We prove that $\begop^*$ is also an optimal solution for
	Problem~\ref{problem:max-cut-variant}.

	Step~(1):~We observe that Problem~\ref{problem:cut-middle-problem} is a
	convex maximization problem over the hypercube, where the objective function
	is a quadratic form with a positive semidefinite matrix.
	Thus the optimal solution must be from the set $\{-1,1\}^n$ (see
	Section~\ref{sec:convexity}).
	Thus, for the remainder of the proof we consider the optimization with
	$\begop \in \{-1, 1\}^n$. 

    Step~(2): We derive a useful characterization of $\efinop{i}$, which we
	obtain through the updating rule of the FJ model. 
    Given an undirected graph $G' = (V, E, \frac{1}{Q}w)$, and innate opinion $\begop$,  
	if we apply the FJ model on this graph $G'$, 
    the updates of the expressed opinions follow the rule
    \begin{equation*}
        \begin{aligned}
            z_i^{(t+1)} &= \frac{s_i + \frac{1}{Q}\sum_j w_{ij}z_j^{(t)}}{1 + \frac{1}{Q}\sum_j w_{ij}} \\
            &= \frac{Q}{Q + \sum_j w_{ij}} s_i + \frac{\sum_j w_{ij}}{Q + \sum_j w_{ij}}\frac{\sum_j w_{ij}z_j^{(t)}}{\sum_j w_{ij}}.
        \end{aligned}
    \end{equation*}
	
	Next, since
	$\finop = \lim_{t \rightarrow \infty} \finop^{(t)}
		= (\ID + \frac{1}{Q}\laplacian)^{-1} \begop$
	is the vector of expressed opinions in the equilibrium, we obtain that for
	all $i$,
	\begin{align*}
		z_i = \delta_i s_i + (1 - \delta_i) \Delta_i,
	\end{align*}
	where we set $\delta_i = \frac{Q}{Q + \sum_j w_{ij}}$, and $\Delta_i =
	\frac{\sum_{j}w_{ij}z_j}{\sum_j w_{ij}}$. 

	Step~(3): We present a technical claim which bounds
	$\finop^{\intercal} \laplacian \finop$ by
	$\begop^{\intercal} \laplacian \begop$ and some small additive terms. We
	prove the claim at the end of the section.
	\begin{claim}
	\label{claim:technical}
		Suppose $\begop\in\{-1,1\}^n$ and let
		$\finop=(\ID+\frac{1}{Q}\laplacian)^{-1} \begop$. Let $M = \sum_{i,j=1} w_{ij}$.
		Then
		$$\begop^{\intercal} \laplacian \begop - \frac{8}{M} - \frac{2}{M^3}
			\leq \finop^{\intercal} \laplacian \finop
			\leq \begop^{\intercal} \laplacian \begop + \frac{4}{M} + \frac{5}{2M^3}.$$
	\end{claim}

    The claim implies that for any graph with sum of weights larger than 8.5,
	i.e., $M \geq 17$, it holds that
	$\begop^{\intercal} \laplacian \begop - 0.5
		< \begop^{\intercal} \laplacian \begop - \frac{8}{M} - \frac{2}{M^3}$ 
    and that
	$\finop^{\intercal} \laplacian \finop
		\leq \begop^{\intercal} \laplacian \begop + \frac{4}{M} + \frac{5}{2M^3}
		< \begop^{\intercal} \laplacian \begop + 0.5$. 
    Thus $[\finop^{\intercal} \laplacian \finop] = \begop^{\intercal} \laplacian \begop$. 
	This proves that the optimal solutions for Problem~\ref{problem:cut-middle-problem}
	and~\ref{problem:max-cut-variant} are identical up to rounding. This also
	implies that Problem~\ref{problem:cut-middle-problem} is \NPhard.
    
	Step~(4): We prove that the optimal solution $\begop^*$ for
	Problem~\ref{problem:cut-middle-problem} is also the optimal 
	solution for Problem~\ref{problem:max-cut-variant}. We prove this by
	contradiction. Assume $\begop^o$ is the optimal solution for
	Problem~\ref{problem:max-cut-variant}, and $\begop^{o\intercal} \laplacian
	\begop^{o} > \begop^{*\intercal} \laplacian \begop^{*}$. 
	Since the feasible areas for Problem~\ref{problem:max-cut-variant} and
	Problem~\ref{problem:cut-middle-problem} are the same, $\begop^{o}$ is also
	a feasible solution for the latter. 
	Let $\finop^o = (\ID + \frac{1}{Q} \laplacian)^{-1} \begop^o$ and observe that since
	$\begop^{o} \in \{-1, 1\}^n$, it holds that  
    $\begop^{o\intercal} \laplacian \begop^o - \frac{8}{M} - \frac{2}{M^2}
		\leq \finop^{o\intercal} \laplacian \finop^o$
	by Claim~\ref{claim:technical}.
    Thus, if $M\geq17$,
   	\begin{align*}
		\finop^{o \intercal} \laplacian \finop^{o}
		&\geq \begop^{o\intercal} \laplacian \begop^o - \frac{8}{M} - \frac{2}{M^3} \\
		&> \begop^{*\intercal} \laplacian \begop^* - \frac{8}{M} - \frac{2}{M^3} \\
		&\geq \begop^{*\intercal} \laplacian \begop^* + 1 - \frac{8}{M} - \frac{2}{M^3} \\
		&\geq \finop^{*\intercal} \laplacian \finop^{*},
	\end{align*}
   where we used that
   $\begop^{o\intercal} \laplacian \begop^o > \begop^{*\intercal} \laplacian \begop^{*}$
   by assumption, that the entries in $\laplacian$ are integers and
   $\begop^{o\intercal}\in\{-1,1\}^n$. In the last step, we applied
   Claim~\ref{claim:technical} and the fact that
   $1 - \frac{8}{M} - \frac{2}{M^3} \geq \frac{4}{M} + \frac{5}{2M^3}$ for $M \geq 13$.
   Note that we can assume that $M\geq 13$ since
   Problem~\ref{problem:max-cut-variant} is \NPhard even for unweighted graphs
   and then $M=\Omega(\abs{E})\gg 13$.
   Since the above inequality contradicts the fact that $\begop^*$ is
   optimal for Problem~\ref{problem:cut-middle-problem}, this finishes the
   proof.

	\begin{proof}[Proof of Claim~\ref{claim:technical}]
	First, let us start by rewriting $\finop^{\intercal} \laplacian \finop$,
	where we use that $\efinop{i} = \delta_i \ebegop{i} + (1 - \delta_i) \Delta_i$.
	Then we have
    \begin{equation}
	\begin{aligned}
	\label{equation:cut-middle-cut-variat-1}
		&\finop^{\intercal} \laplacian \finop
		= \sum_{i,j=1}^n \efinop{i} \efinop{j} \laplacian_{ij} \\
		&= \sum_{i,j=1}^n (\delta_i \ebegop{i} + (1 - \delta_i) \Delta_i) (\delta_j \ebegop{j} + (1 - \delta_j) \Delta_j) \laplacian_{ij} \\
		&= \sum_{i,j=1}^n \big[\delta_i \delta_j \ebegop{i} \ebegop{j} + (1 - \delta_i)\Delta_i \delta_j \ebegop{j} + \delta_i \ebegop{i} (1 - \delta_j) \Delta_j \\
		&\quad\quad\quad\quad + (1-\delta_i)(1 - \delta_j)\Delta_i\Delta_j \big] \laplacian_{ij}.
    \end{aligned}
	\end{equation}

	In the following, we split
	Equation~\eqref{equation:cut-middle-cut-variat-1} into
	$\sum_{i,j=1}^n \delta_i \delta_j \ebegop{i} \ebegop{j} \laplacian_{ij}$
	and $\sum_{i,j=1}^n ((1 - \delta_i)\Delta_i \delta_j \ebegop{j} + \delta_i
			\ebegop{i} (1 - \delta_j) \Delta_j + (1-\delta_i)(1 -
				\delta_j)\Delta_i\Delta_j) \laplacian_{ij}$. 
	For both terms, we will derive upper and lower bounds.
	However, before we do that we first prove bounds for $\delta_i$, $\Delta_i$
	and for $\begop^{\intercal} \laplacian \begop$.
      
	Now we derive bounds for $\delta_i$:
	Let $D_{ii} = \sum_{j=1}w_{ij}$. If we pick $Q \geq (M^2 - 1) \max_i D_{ii}$
	and using the definition $\delta_i = \frac{Q}{Q+\sum_j w_{ij}}$, it follows that 
	for any $i$, $\delta_i \leq 1$ and
    \begin{equation}
        \label{equation:cut-middle-cut-variat-1-2}
        \begin{aligned}
            \delta_i 
			\geq \frac{(M^2 - 1) \max_j D_{jj}}{(M^2 - 1) \max_j D_{jj} + D_{ii}}
			\geq \frac{(M^2 - 1) D_{ii}}{(M^2 - 1)D_{ii} + D_{ii}} = \frac{M^2 - 1}{M^2}.
        \end{aligned}
    \end{equation}
	
    Next, we argue that $ -1 \leq \Delta_i \leq 1$:
	By definition of $\Delta_i$ the claim follows if we show that
	$-1\leq z_j^{(t)}\leq 1$ for all $t$ and $j$.
	For $t=0$ we have that $\efinop{j}^{(0)} = \ebegop{j}$ and the claim
	is true.
	Now observe that the update rule for $\efinop{i}^{(t+1)}$ implies that at
	each iteration it holds that,
	$\min \{s_i, z_j^{(t)} \mid j \neq i\}
		\leq \efinop{i}^{(t+1)}
		\leq \max \{s_i, z_j^{(t)} \mid j \neq i\}$.
	Now by induction we obtain that each $\efinop{i}$ is always upper bounded by
	$\max\{s_i \mid i = 1, 2, \ldots, n\}$ 
    and lower bounded by $\min \{s_i \mid i = 1, 2, \ldots, n\}$.
	Since we have that $-1\leq\begop_i\leq1$ for all $i$, we obtain our result
	for $\Delta_i$.
    
	Next, we obtain the following bounds on
	$\begop^{\intercal} \laplacian \begop$ for $\begop \in \{-1, 1\}^n$:
	\begin{equation}
        \label{equation:cut-middle-cut-variat-1-3}
        \begin{aligned}
            0 \leq \begop^{\intercal} \laplacian \begop
			\leq \sum_{i,j} \abs{\laplacian_{i, j}}
			= 2M,
        \end{aligned}
    \end{equation}
	where we used that $\laplacian$ is positive semidefinite and that, since
	$G$ is an undirected graph, $M$ is twice the sum of the edge weights.
   
	Given the fact that $\ebegop{i}$ is equal to either $1$ or $-1$, and
	Equation~\eqref{equation:cut-middle-cut-variat-1-2}~\eqref{equation:cut-middle-cut-variat-1-3},
	it follows that 
    \begin{equation}
        \label{equation:cut-middle-cut-variat-2}
        \begin{aligned}
            \sum_{i,j=1}^n \delta_i \delta_j \ebegop{i} \ebegop{j} \laplacian_{ij}
			&= \frac{1}{2} \sum_{i,j=1}^n w_{ij} (\delta_i \ebegop{i} - \delta_j \ebegop{j})^2  \\
            &\geq \frac{1}{2} \sum_{i,j=1, \ebegop{i} \neq \ebegop{j}}^n w_{ij} (\delta_i + \delta_j)^2  \\
            &\geq \frac{1}{2} \sum_{i,j=1, \ebegop{i} \neq \ebegop{j}}^n w_{ij} \left(\frac{M^2 - 1}{M^2}\right)^2 2^2  \\
            &= \frac{1}{2} \sum_{i,j=1}^n w_{ij} \left(\frac{M^2 - 1}{M^2}\right)^2(\ebegop{i} - \ebegop{j})^2  \\
            &= \left(\frac{M^2 - 1}{M^2}\right)^2 \begop^{\intercal} \laplacian \begop \\
			&\geq \begop^{\intercal} \laplacian \begop - \frac{4}{M},
        \end{aligned}
    \end{equation}
	where the last inequality comes from
	Equation~\eqref{equation:cut-middle-cut-variat-1-3} which implies that 
    $\left(\frac{M^2 - 1}{M^2}\right)^2 \begop^{\intercal} \laplacian \begop 
		\geq \begop^{\intercal} \laplacian \begop - \frac{2}{M^2}\begop^{\intercal} \laplacian \begop
		\geq \begop^{\intercal} \laplacian \begop - \frac{4}{M}$.  

	Similarly, in Equation~\eqref{equation:cut-middle-cut-variat-3} we prove an
	upper bound on
	$\sum_{i,j=1}^n \delta_i \delta_j \ebegop{i} \ebegop{j} \laplacian_{ij}$,
    \begin{equation}
        \label{equation:cut-middle-cut-variat-3}
        \begin{aligned}
            &\sum_{i,j=1}^n \delta_i \delta_j \ebegop{i} \ebegop{j} \laplacian_{ij}
			= \frac{1}{2} \sum_{i,j=1}^n w_{ij} (\delta_i \ebegop{i} - \delta_j \ebegop{j})^2  \\
			&= \frac{1}{2}\sum_{i,j=1, \ebegop{i} \neq \ebegop{j}}^n w_{ij}(\delta_i + \delta_j)^2 + \frac{1}{2} \sum_{i,j=1, \ebegop{i} = \ebegop{j}}^n w_{ij}(\delta_i - \delta_j)^2 \\
            &\leq  \frac{1}{2}\sum_{i,j=1, \ebegop{i} \neq \ebegop{j}}^n w_{ij} 2^2
				+ \frac{1}{2}\sum_{i,j=1}^n w_{ij} \left(\frac{1}{M^2}\right)^2 \\
            &\leq \begop^{\intercal} \laplacian \begop + \frac{1}{2}\sum_{i,j=1}^n w_{ij} \frac{1}{M^4} \\
            &= \begop^{\intercal} \laplacian \begop + \frac{1}{2M^3},
        \end{aligned}
    \end{equation}
    where in the second line we used that
	$\delta_i + \delta_j \leq 2$ and
	$\abs{\delta_i - \delta_j} \leq 1 - \frac{M^2 - 1}{M^2} \leq \frac{1}{M^2}$,
	and in the final equality we used that $M$ is twice the sum over all edge weights.

	We continue to bound the remaining part of $\finop^\intercal \laplacian
	\finop$, first we get an upper bound in
	Equation~\eqref{equation:cut-middle-cut-variat-4}, 
     \begin{equation}
        \label{equation:cut-middle-cut-variat-4}
        \begin{aligned}
            &\sum_{i,j=1}^n ((1 - \delta_i)\Delta_i \delta_j \ebegop{j} + \delta_i \ebegop{i} (1 - \delta_j) \Delta_j + (1-\delta_i)(1 - \delta_j)\Delta_i\Delta_j) \laplacian_{ij} \\
            &\leq \sum_{i,j=1}^n ((1-\delta_i) + (1 - \delta_j) + (1 - \delta_i)(1 - \delta_j))\abs{\laplacian_{ij}} \\
            &\leq \sum_{i,j=1}^n \left(\frac{1}{M^2} +\frac{1}{M^2} + \frac{1}{M^4}\right)\abs{\laplacian_{ij}} = \frac{4}{M} + \frac{2}{M^3},
        \end{aligned}
    \end{equation}
	where the first inequality holds since $\abs{\Delta_i} \leq 1$,
	$\abs{\delta_i} \leq 1$ and $\abs{s_i} \leq 1$. 
    The second inequality holds since $\delta_i \geq \frac{M^2 - 1}{M^2}$, and 
    $\sum_{i,j}\abs{\laplacian_{ij}} = 2M$.
	
	Similarly, we give a lower bound in
	Equation~\eqref{equation:cut-middle-cut-variat-5}:
     \begin{equation}
        \label{equation:cut-middle-cut-variat-5}
        \begin{aligned}
            &\sum_{i,j=1}^n ((1 - \delta_i)\Delta_i \delta_j \ebegop{j} + \delta_i \ebegop{i} (1 - \delta_j) \Delta_j + (1-\delta_i)(1 - \delta_j)\Delta_i\Delta_j) \laplacian_{ij} \\
            &\geq -\sum_{i,j=1}^n ((1-\delta_i) + (1 - \delta_j) + (1 - \delta_i)(1 - \delta_j))\abs{\laplacian_{ij}} \\
            &\geq -\sum_{i,j=1}^n \left(\frac{1}{M^2} +\frac{1}{M^2} + \frac{1}{M^4}\right)\abs{\laplacian_{ij}} = -\frac{4}{M} - \frac{2}{M^3}.  
        \end{aligned}
    \end{equation}

    Combining all together, we conclude
	$\begop^{\intercal} \laplacian \begop - \frac{8}{M} - \frac{2}{M^3} 
		\leq \finop^{\intercal} \laplacian \finop
		\leq \begop^{\intercal} \laplacian \begop + \frac{4}{M} + \frac{5}{2M^3}$.
	\end{proof}

\begin{proof}[Proof of Theorem~\ref{thm:disagreement-np-hard}]
	We prove that Problem~\eqref{problem:max-disagreement-without-constraint} is \NPhard.
	Notice that since the objective function of the problem is convex (see
	Section~\ref{sec:convexity}) the optimal
	solution~$\begop^*$ for Problem~\eqref{problem:max-disagreement-without-constraint} will be a
	vector with all entries in the set $\{-1,1\}$.

	Consider an input $(G=(V,E,w),Q)$ for
	Problem~\eqref{problem:cut-middle-problem}. We use an algorithm for
	Problem~\eqref{problem:max-disagreement-without-constraint} and apply it on the graph with edge
	weights $\frac{w}{Q}$, i.e., each edge~$e$ has weight~$\frac{w_e}{Q}$.
	We note that for Problem~\eqref{problem:max-disagreement-without-constraint} we do not assume
	that the edge weights are integers.

	We let $\OPTONE(V, E, \frac{w}{Q})$ denote the optimal objective value
	for Problem~\eqref{problem:max-disagreement-without-constraint} on input $G = (V, E, \frac{w}{Q})$.  
	We have:
	\begin{align*}
		Q\cdot \OPTONE(V, E, \frac{1}{Q}w) 
		&= Q \! \max_{\begop\in\{-1,1\}^n}
				\begop^{\intercal}
				(\ID + \frac{1}{Q} \laplacian)^{-1}
				\frac{1}{Q} \laplacian
				(\ID + \frac{1}{Q} \laplacian)^{-1}
				\begop \\
		&= \max_{\begop\in\{-1,1\}^n,\finop=(\ID + \frac{1}{Q} \laplacian)^{-1}\begop}
				\finop^{\intercal} \laplacian \finop \\
		&= \OPTTHREE(V, E, w).
	\end{align*}

	We can now apply
	Lemma~\ref{lemma:cut-middle-cut-variat} and it follows that
	Problem~\eqref{problem:max-disagreement-without-constraint} is~\NPhard.  

	The fact that Problem~\eqref{problem:max-disagreement} and Problem~\eqref{eq:our-problem} 
	are \NPhard follow immediately
	from the result, and the analysis from the beginning of the section.
\end{proof}

\subsection{Proof of Corollary~\ref{cor:our-problem-np-hard}}

Given that Problem~\eqref{problem:max-disagreement-without-constraint} is \NPhard,
we proceed to prove that Problem~\eqref{eq:our-problem} is \NPhard when $k = \frac{n}{2}$.
We denote the latter problem with cardinality constraint $k = \frac{n}{2}$ as Problem~$C$. 
We apply the same technique that shows that Bisection is \NPhard
given that {\sc Max\-Cut} is \NPhard~\cite{garey1974some}. 

Consider an instance of
Problem~\eqref{problem:max-disagreement-without-constraint} with $n'=k$ vertices, 
where the disagreement index matrix is represented as 
$\DisIdx{\laplacian} = (\ID + \laplacian)^{-1} \laplacian (\ID + \laplacian)^{-1}$. 
We construct an instance of Problem $C$ by adding $n'$ isolated nodes. Note
that in the new instance we have $2n' = n$ vertices and $k=\frac{n}{2} = n'$.

The Laplacian of the graph after adding the isolated nodes becomes
$\laplacian' =
\begin{bmatrix}
\m+0 & \m+0 \\
\m+0 & \laplacian
\end{bmatrix}
$. Thus, we can express the disagreement index matrix 
of our new instance for Problem~$C$ as:
\begin{displaymath}
\begin{aligned}
\DisIdx{\laplacian'}
	&= \begin{bmatrix}
			\m+I & \m+0 \\
			\m+0 & \ID + \laplacian
		\end{bmatrix}^{-1}
		\begin{bmatrix}
			\m+0 & \m+0 \\
			\m+0 & \laplacian
		\end{bmatrix}
		\begin{bmatrix}
			\m+I & \m+0 \\
			\m+0 & \ID + \laplacian
		\end{bmatrix}^{-1} \\
	&= \begin{bmatrix}
			\m+I & \m+0 \\
			\m+0 & (\ID + \laplacian)^{-1}
		\end{bmatrix}
		\begin{bmatrix}
			\m+0 & \m+0 \\
			\m+0 & \laplacian
		\end{bmatrix}
		\begin{bmatrix}
			\m+I & \m+0 \\
			\m+0 & (\ID + \laplacian)^{-1}
		\end{bmatrix} \\
	&= \begin{bmatrix}
			\m+0 & \m+0 \\
			\m+0 & (\ID + \laplacian)^{-1}\laplacian(\ID + \laplacian)^{-1}
		\end{bmatrix} \\
	&= \begin{bmatrix}
			\m+0 & \m+0 \\
			\m+0 & \DisIdx{\laplacian}
		\end{bmatrix},
\end{aligned}
\end{displaymath}
i.e., $\DisIdx{\laplacian'}$ can be partitioned into four blocks, where
each block is an $n' \times n'$ matrix, with three of them being $\mathbf{0}$ matrices. 
The remaining block is the $n' \times n'$ matrix $\DisIdx{\laplacian}$ from our
instance of Problem~\eqref{problem:max-disagreement-without-constraint}.

Now consider the solution~$\begop$ of our instance $\DisIdx{\laplacian'}$ for
Problem~$C$. Observe that the solution's objective function value is given by
$\begop^\intercal \DisIdx{\laplacian'} \begop$. Note that, due to the
block structure of $\DisIdx{\laplacian'}$, only the final $n'$ entries of
$\begop$ create a non-negative disagreement. Furthermore, the first $n'$
entries of $\begop$ do not contribute to the final disagreement and can be set
arbitrarily.

Hence, since our cardinality constraint is $k=n'$, the algorithm can pick the
solution for the final $n'$~vertices such as to maximize the disagreement
$\DisIdx{\laplacian}$ on our instance of
Problem~\eqref{problem:max-disagreement-without-constraint}; the first
$n'$~entries of $\begop$ can be picked such that the cardinality constraint is
satisfied.  Therefore, the optimal objective function values of our new instance
and the original instance for
Problem~\eqref{problem:max-disagreement-without-constraint} are identical.

Consequently, we conclude that
Problem~\eqref{problem:max-disagreement-without-constraint} can be solved by
utilizing a solver for 
Problem~$C$, and the reduction is in polynomial time.
 
\balance

\end{document}